\documentclass[a4paper, 11p]{amsart}
\newtheorem{definition}{Definition}
\newtheorem{proposition}{Proposition}
\newtheorem{theorem}{Theorem}
\newtheorem{corollary}{Corollary}

\usepackage{graphicx}
\usepackage{amsmath,amssymb,latexsym,bbm}
\newcommand{\beq}{\begin{equation}}
\newcommand{\eeq}{\end{equation}}
\newcommand{\beqa}{\begin{eqnarray}}
\newcommand{\eeqa}{\end{eqnarray}}

\begin{document}

\title{Combinatorial aspects of the Sachdev-Ye-Kitaev model}

\author{M. Laudonio \and R. Pascalie \and A. Tanasa}

\maketitle

\begin{abstract}
The Sachdev-Ye-Kitaev (SYK) model is a model of $q$ interacting fermions whose large N limit is dominated by melonic graphs. In this review we first present a diagrammatic proof of that result by direct, combinatorial analysis of its Feynman graphs. Gross and Rosenhaus have then proposed a generalization of the SYK model which involves fermions with different flavors. In terms of Feynman graphs, these flavors can be seen as reminiscent of the colors used in random tensor theory. Applying modern tools from random tensors to such a colored SYK model, all leading and next-to-leading orders diagrams of the 2-point and 4-point functions in the large $N$ expansion can be identified. We then study the effect of non-Gaussian average over the random couplings in a complex, colored version of the SYK model. Using a Polchinski-like equation and random tensor Gaussian universality, we show that the effect of this non-Gaussian averaging  leads to a modification of the variance of the Gaussian distribution of couplings at leading order in $N$. We then derive the form of the effective action to all orders.

%Insert your abstract here. Include keywords, PACS and mathematical
%subject classification numbers as needed.
%\and More}
%\PACS{PACS code1 \and PACS code2 \and more}
%\subclass{MSC code1 \and MSC code2 \and more}
\end{abstract}

\tableofcontents

\newpage
%%%%%%%%%%%%%
\section{Introduction} \label{intro}
%%%%%%%%%%%%%

The Sachdev-Ye model \cite{sachdev} was introduced within a condensed matter framework in the early nineties. This model attracted a certain interest within the condensed matter community. Thus, the 2-point function computation in the large $N$ limit was performed in \cite{georges}. In a series of talks \cite{kitaev}, Kitaev
introduced a simplified  version of this model and showed
it can be a particularly interesting toy-model for AdS/CFT physics. The model, called ever since the Sachdev-Ye-Kitaev (SYK) model, has attracted a huge amount of interest for both condensed matter and high energy physics, see for example \cite{maldacena}, \cite{polchinski}, \cite{Gross}, or the review articles \cite{Sarosi} or \cite{Rosenhaus}.

More specifically, the SYK model is a quantum-mechanical model with $N$ fermions with random interactions involving $q$ of these fermions at a time. Each coupling $J$ is a variable drawn from a random Gaussian distribution. From a theoretical phyisc point of view, the model has three remarkable properties: it is solvable at strong coupling, maximally chaotic and presents an emergent conformal symmetry both spontaneously and explicitly broken.

%The SYK model is the only toy-model known so far to enjoy all these properties (other known models have some of these properties but no other model has all of them). 
%It is this fact which led to the huge amount of interest within the high energy physics community.

A crucial property for the above mentioned solvability of the SYK model is that it is dominated by melonic graphs in the large $N$ limit. Remarkably, those graphs had been known to dominate the large $N$ limit of random tensor models \cite{gurau-book}, $N$ being here the size of the tensor. This is true for the colored tensor model \cite{continuum_TM}, the multiorientable model \cite{DRT}, \cite{io-mo}, the $O(N)^3$-invariant model \cite{ct} and so on. Even for tensor models whose large $N$ limit does no consist of melonic graphs, the universality class of melonic graphs is easily stumbled upon \cite{bonzomr}. 

The large $N$ dominance of melonic graphs in both the SYK model and tensor models triggered interesting developments, starting from the Gurau-Witten model \cite{Witten} and the model $O(N)^3$-invariant model \cite{CTKT}, which are reformulations of the SYK model using fermionic tensor fields without quenched disorder. This has motivated $1/N$ expansions for new tensorial models \cite{IrreducibleTensorsCarrozza}, \cite{SymmetricTraceless}, \cite{TensorialGrossNeveu}, \cite{TwoSymmetricTensors} and the new field of tensor quantum mechanics \cite{SpectraTensorModels}, \cite{Sextic}, \cite{Carrozza:2018psc}, \cite{Benedetti:2019eyl}.

We also consider in this review paper a version of the SYK model containing $q$ flavors (or colors) of complex fermions, each of them appearing once in the interaction. This model is very close in the spirit to the colored tensor model (see the book \cite{gurau-book}) and it is a particular case of a complex version of the Gross-Rosenhaus SYK generalization proposed in \cite{Gross}. This particular version of the SYK model has already been studied in \cite{complete}, \cite{gurau-ultim}, \cite{DartoisCFT} and \cite{Fusy:2018xax}.

In this model and in the Gurau-Witten model, combinatorial methods originating from the study of tensor models can be used, for instance to identify the Feynman graphs which contribute at a given order in the $1/N$ expansion \cite{complete}.
However, combinatorial proofs have been of limited use so far in the original SYK model. One reason is that colors (or flavors) have been crucial to most combinatorial results in tensor models but the Feynman graphs of the SYK model have no colors (only recently new methods have been found to deal with tensor models without colors \cite{IrreducibleTensorsCarrozza}, \cite{SymmetricTraceless}, \cite{TensorialGrossNeveu}, \cite{TwoSymmetricTensors}, but have not been applied so far to the original SYK model). 
%to the best of our knowledge).

In this review, we first follow \cite{Bonzom:2018jfo} and present the proof of the dominance of melonic graphs of the SYK model, see Theorem \ref{thm} below. 
%One might expect it to be difficult without the notion of colors. As it turns out, even for graphs without colors, the set of melonic graphs is simple enough that our direct, combinatorial approach remains quite elementary. It revolves around the fact that ultimately, melonic graphs are characterized by their 2-cuts, so we can study the large $N$ behavior of Feynman graphs under 2-cuts.
We then describe the study done in \cite{Bonzom:2017pqs}, of the real colored SYK model using some elementary version of very recent combinatorial techniques for edge-colored graphs \cite{blr}. Those techniques have been developed in order to go (successfully to some extent \cite{bonzomr}) beyond the melonic phase in ordinary tensor models. We will explain here how to extract the leading order (LO), the next-to-leading order (NLO) of the 2-point and 4-point functions using the most simple version of those techniques. The LO reproduces trivially melons and chains (known as ladders in \cite{maldacena}), and give new graphs at NLO. We also included some details on the next-to-next-to-leading order (NNLO) of the 2-point function to show that this method is fairly straightforward to apply. 

We then  consider a complex colored SYK model with non-Gaussian disorder. Following the approach proposed in \cite{Krajewski} for tensor models and group field theory (see also \cite{Krajewski2}, \cite{Krajewski3} and \cite{Krajewski4}),  
we first use a Polchinski-like flow equation to obtain Gaussian universality. This Gaussian universality result for the colored tensor model was initially proved in \cite{universality}. Let us also mention here that this universality result for colored tensor models was also exploited in \cite{Bonzom}, in a condensed matter physics setting.
%, to identify an infinite universality class of infinite-range $p-$spin glasses with non-Gaussian correlated quenched distributions. 
We further obtain the effective action of the model and show that the effect of the non-Gaussian disorder is a modification of the variance of the Gaussian distribution of couplings at leading order in $N$.

This review paper is organized as follows. In the following section, we introduce the SYK model, its Feynman graphs and we give the proof of the melonic dominance of the SYK model. In subsection \ref{sec:defcolored} we first recall the Gross-Rosenhaus generalization of the SYK model and specify the versions we will study. 
In the following subsection, we exhibit what the LO and NLO vacuum, two- and four-point diagrams are, using a simple method alternative to \cite{gurau-ultim}. In subsection \ref{sec:nongaus} we express the non-Gaussian potential as a sum over particular graphs and show the Gaussian universality using a Polchinski-like equation. We then study the effective action of the model. Finally, in the last section, we present some concluding remarks and perspectives.

Throughout the text, Feynman graphs are represented with $q=4$, $q$ being here the number of fermions interacting at each vertex. 

%We only discuss the diagrammatics of the $1/N$ expansions and leave out completely the Feynman amplitudes.

\section{The Sachdev-Ye-Kitaev model}
\label{sec:SYK}

\subsection{Definition of the model and its Feynman graphs}
\label{sec:Def}

%The SYK model contains $N$ real fermions $\psi_i$ ($i=1,\ldots, N$), with $q-$fold random coupling, $q$ being here an even integer.
%The action writes:
%\beqa
%\label{act:syk}
%S_{\mathrm{SYK}}=\int d\tau \left( 
%\frac 12 \sum_{i=1}^N \psi_i \frac{d} {dt} \psi_i
%- \frac{i^{q/2}}{q!}
%\sum_{i_1,\ldots, i_q=1}^N
%j_{i_1\ldots i_q}\psi_{i_1}\ldots \psi_{i_q}
%\right)
%\eeqa
%Note that the most widely studied version of the SYK model has real fermions, as above. This was done in \cite{kitaev} and subsequent papers.

The SYK model is a qunatum mechanical model with $N$ Majorana fermions 
$\psi_i$ ($i=1,\ldots, N$)
coupled via a $q$-body random interaction ($q$ being here an even integer)
\beqa
\label{act:syk}
S_{\mathrm{SYK}}=\int d\tau \left( 
\frac 12 \sum_{i=1}^N \psi_i \frac{d} {dt} \psi_i
- \frac{i^{q/2}}{q!}
\sum_{i_1,\ldots, i_q=1}^N
j_{i_1\ldots i_q}\psi_{i_1}\ldots \psi_{i_q}
\right)
\eeqa
%\begin{equation}
%H_{{\mathrm{SYK}}} = i^{q/2} \sum_{i_1, \dotsc, i_q} J_{i_1\dotsb i_q} \psi_{i_1}(t) \dotsm \psi_{i_q}(t),
%\label{H_SYK}
%\end{equation}
where $J_{i_1\dotsb i_q}$ is the coupling constant. 
Furthermore, the model is quenched, which by definition means that 
the coupling $J$ 
%Quenching means that it 
is a random tensor with a Gaussian distribution such that 
\begin{equation} \label{Covariance}
\langle J_{i_1\dotsb i_q}\rangle=0 \qquad \text{and} \qquad \langle J_{i_1\dotsb i_q} J_{j_1\dotsb j_q}\rangle = (q-1)! J^2 N^{-(q-1)} \prod_{m=1}^{q} \delta_{i_m,j_m}.
\end{equation}
The fields $\psi_i(t)$ satisfy fermionic anticommutation relations $\{\psi_i(t),\psi_j(t)\} = \delta_{i,j}$. This anticommutation property excludes graphs with tadpoles (also known as loops, in a graph theoretical language); the model being $(0+1)-$dimensional, the Feynman amplitude of such a graph is zero.

\medskip

%%%%%%%%%%%%
%\subsection{Stranded structure of the Feynman graphs}
%\label{sec:Stranded}

In a Feynman graph of the SYK model, the interaction term  is represented by a vertex with $q$ incident fermionic lines. Each fermionic edge $m=1,\dotsc, q$ carries an index $i_m=1, \dotsc, N$ which is contracted at the vertex with a coupling constant $J_{i_1 \dotsb i_q}$. 
The free energy expands onto those connected, $q$-regular and tadpoleless (or loopless, in a graph theoretical language) graphs.

The so-called average over the disorder is done using standard QFT Wick contractions between pairs of $J$s, with covariance \eqref{Covariance}. An additional edge is thus 
adjacent to each vertex. We represent this additional edge as a dashed edge and we call it a {\it disorder edge}. An example of such a Feynman graph of the SYK model is given  in Fig.~\ref{SYK_graph}.

\begin{figure}[ht]
\centering
\includegraphics[scale=0.65]{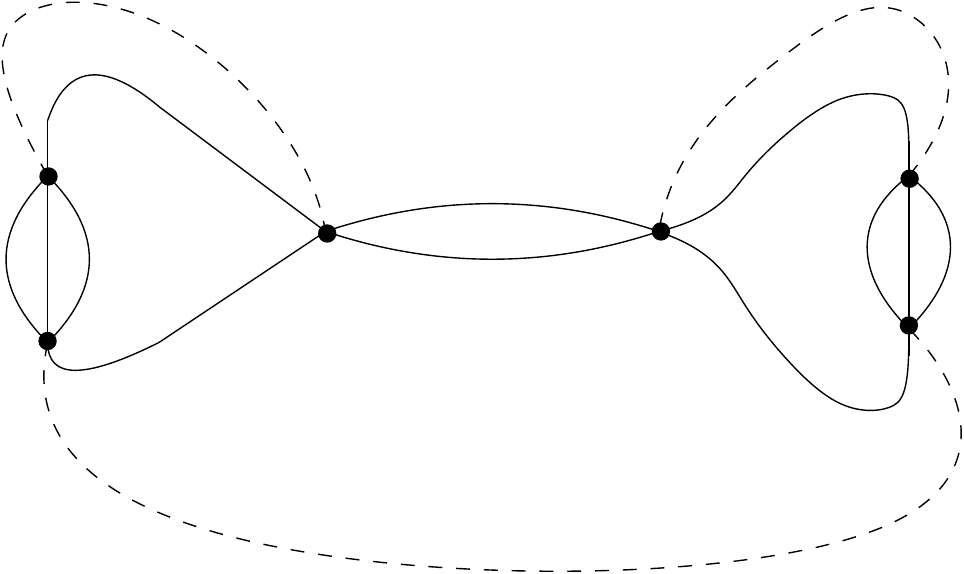}
\caption{An example of a Feynman graph of the $q=4$ SYK model.}
\label{SYK_graph}
\end{figure}

The above description of the Feynman graphs 
%is however not enough to describe the $1/N$ expansion as it 
ignores the indices of the random couplings. Indeed, a disorder edge propagates  $q$ field indices, where the field index of fermionic line incident on a vertex is identified with the index of a fermionic line at another vertex. We thus %have to 
represent a disorder edge as an edge made of $q$ strands, where each strand connects fermionic edges as follows:
\begin{equation}
\langle J_{i_1\dotsb i_q} J_{j_1\dotsb j_q}\rangle = \begin{array}{c} \includegraphics[scale=.4]{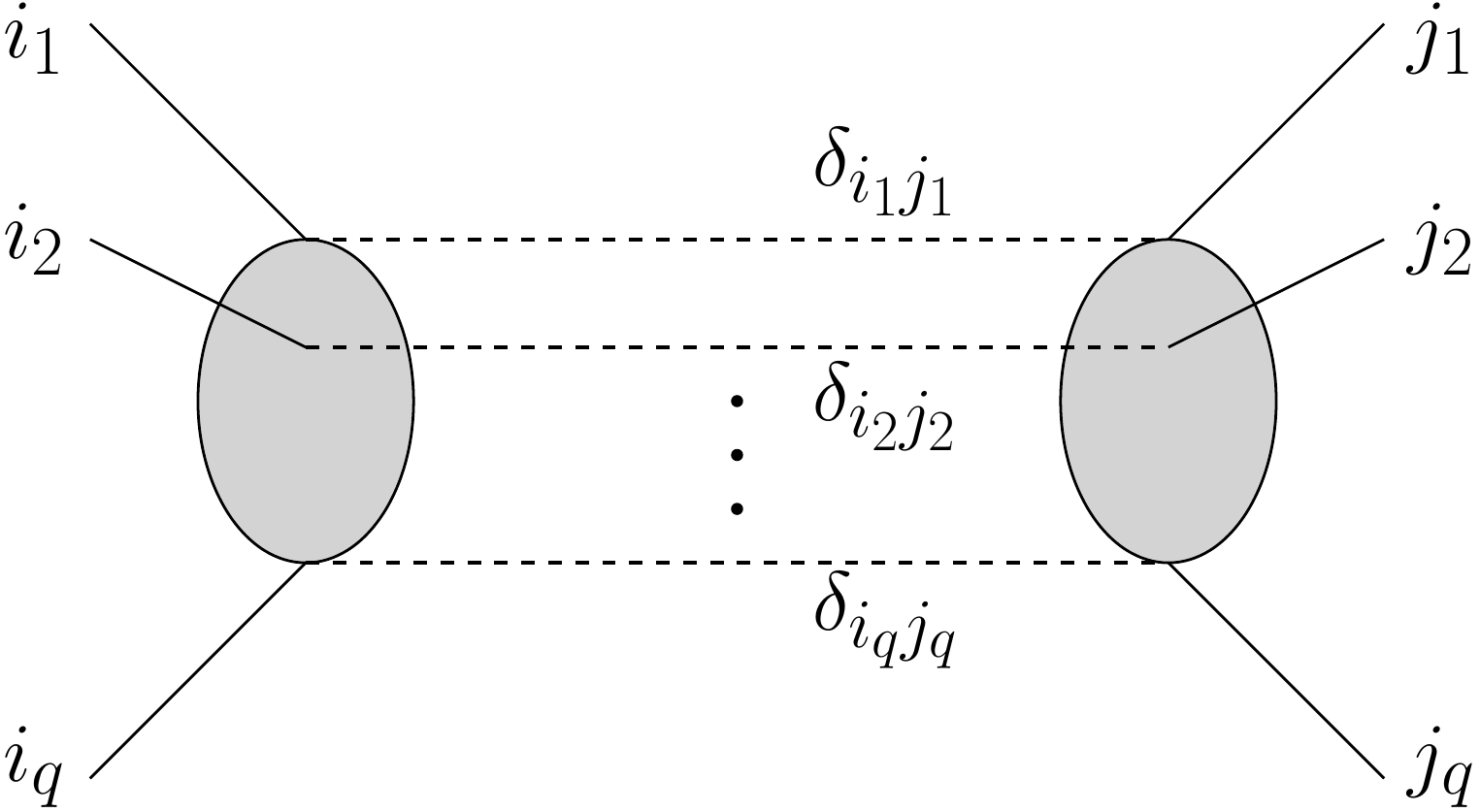} \end{array}
\end{equation}
Here the grey discs represent the Feynman vertices.

We denote by $\mathbbm{G}$ the set of Feynman graphs of the SYK model. 
%with stranded disorder lines. 
For $G\in \mathbbm{G}$, we further denote $G_0\subset G$ the $q$-regular graph obtained by removing the strands of the disorder lines, see Fig.~\ref{SYK_graph0}. 
%D
%ue to the quenching averaging the fermionic free energy over the disorder, 
Note that the graph $G_0$ has to be connected. Moreover, 
%due to the Wick contractions, 
each vertex of the graph $G$ has exactly one adjacent disorder edge. This implies that the graphs $G$ and $G_0$ have an even number of vertices.
 
\begin{figure}[ht]
\centering
\includegraphics[scale=0.75]{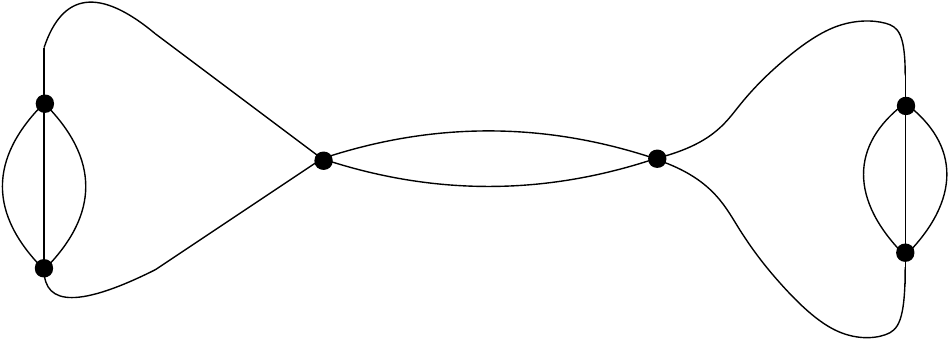}
\caption{The graph obtained after deleting of the disorder lines of the graph of Fig.~\ref{SYK_graph}.}
\label{SYK_graph0}
\end{figure}

Let us now give the following definition:

\begin{definition}
%[Faces]
A cycle made of alternating fermionic lines and strands of disorder lines is called a face. We denote $F(G)$ the number of faces of $G\in\mathbbm{G}$.
\end{definition}

Let us consider graphs with two vertices and $G_{0,\min} = \begin{array}{c} \includegraphics[scale=.25]{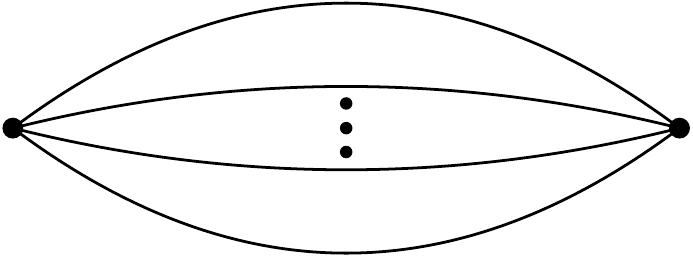} \end{array}$. There are $q!$ such graphs corresponding to permutations of the strands of the disorder line connecting these two vertices. However, among these $q!$ graphs there is only one which maximizes the number of vertices.
%There is a single graph with two vertices
This graph is:
\begin{equation}
G_{\min} = \begin{array}{c} \includegraphics[scale=.45]{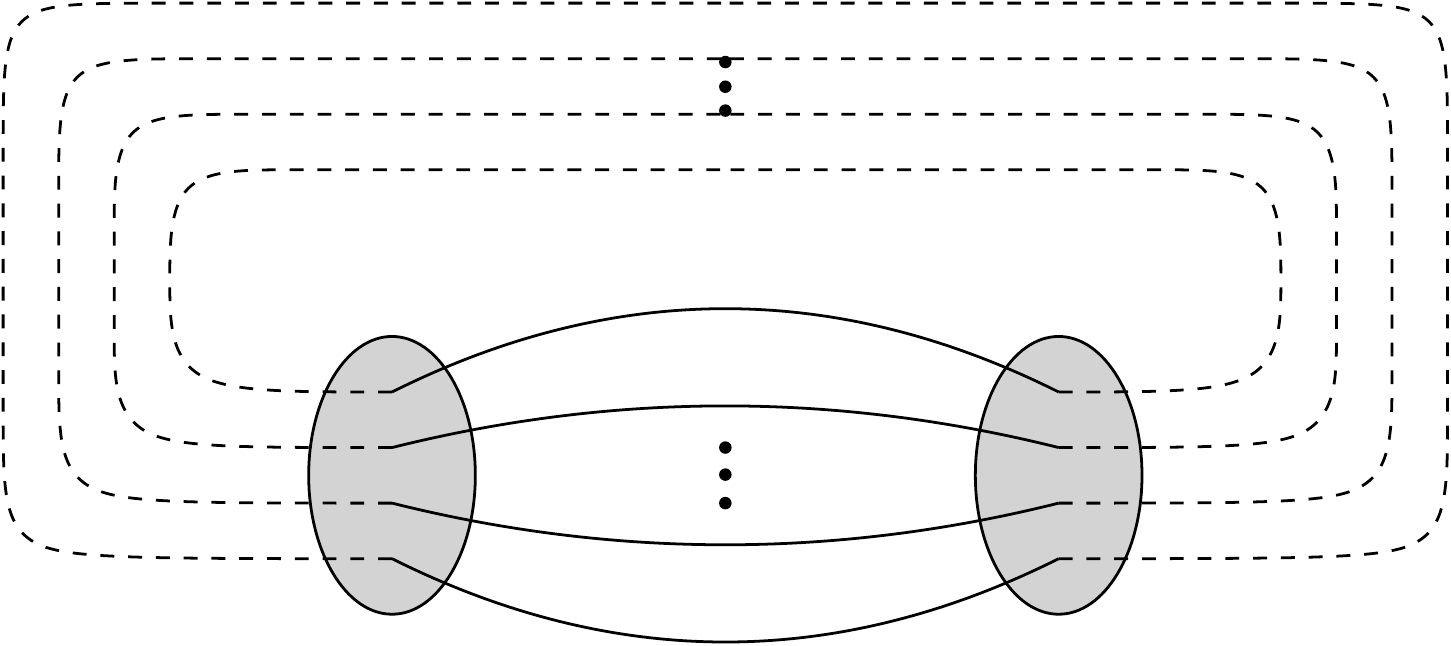} \end{array}
\end{equation}
%and $G_{0,\min} = \begin{array}{c} \includegraphics[scale=.25]{2VertexGraph0SYK.pdf} \end{array}$.

%The field index is preserved along each fermionic edge and along each strand of the disorder edges. This means that there is a free sum for each cycle made of those edges, thereby contributing to a factor $N$.

When computing the Feynman amplitude of an SYK graph in the large $N$ limit, there is a contribution of 
%$G\in\mathbbm{G}$ thus receives 
a factor $N$ per face. The Feynman amplitude also receives a factor $N^{-(q-1)}$ for each disorder edge. The
large $N$ limit Feynman amplitude of an SYK graph is thus
$$N^{\delta(G)} $$
where
\begin{equation} \label{scaling}
%W(G) = 
\delta(G) = F(G) - (q-1) V(G)/2
\end{equation}
where $V(G)$ is the number of vertices. 
We call the parameter $\delta (G)$ the {\bf SYK degree} of the Feynman graph $G$.
To find the 
dominant graphs in this
large $N$ limit, we thus need to find the graphs which maximize the number of faces at fixed number of vertices.

\subsection{Diagrammatic proof of the large $N$ melonic dominance}
\label{sec:proofmelon}

This subsection follows the original article \cite{Nador}. Let us first give the following definitions:

\begin{definition}
We call {\bf dipole} the following 2-point graph:
\begin{equation}
D = \begin{array}{c} \includegraphics[scale=0.4]{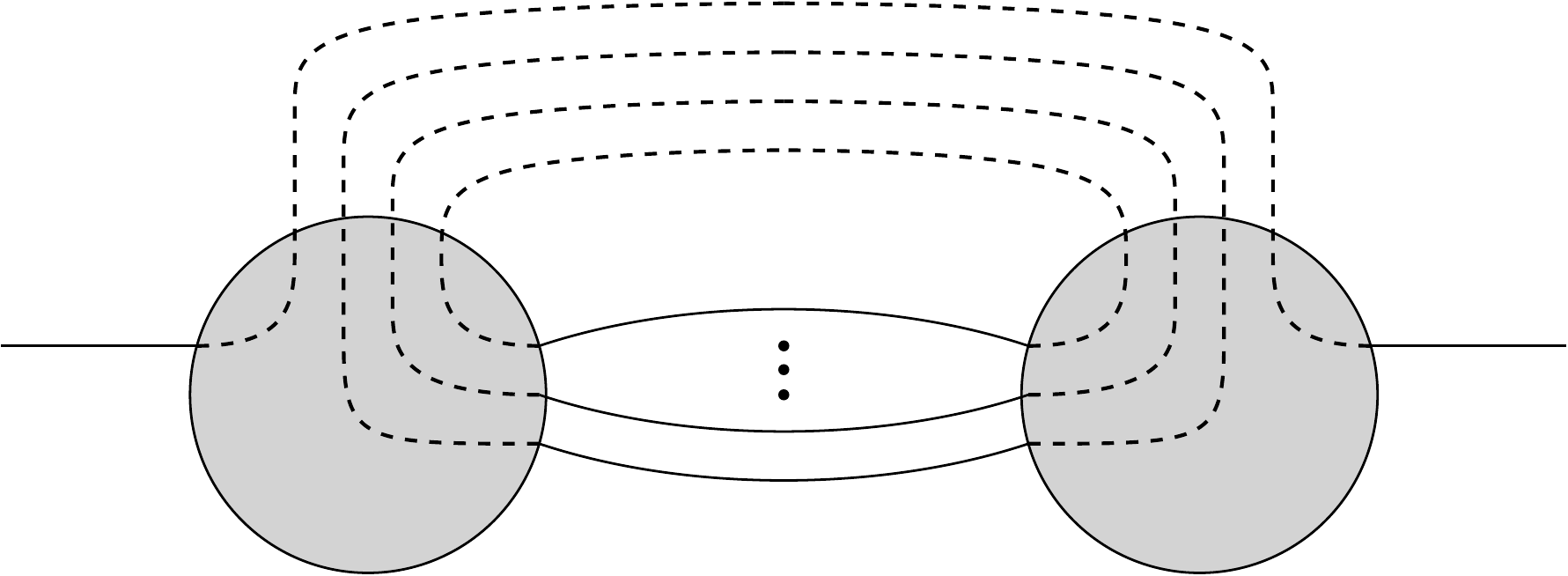} \end{array}
\end{equation}
\end{definition}

Let us note that a dipole is made of two vertices connected by $(q-1)$ fermionic lines and a disorder line. {\it A priori}, there are $q!$ ways of connecting the strands of the disorder line. Among these $q!$ possibilities, we chose for $D$ the one  which creates the maximal number of faces.

\begin{definition}
\label{def:melonicmove}
A {\bf melonic move} is the insertion of a dipole on a fermionic line:
\begin{equation}
\begin{array}{c} \includegraphics[scale=.4]{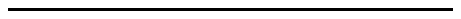} \end{array} \qquad \to \qquad \begin{array}{c} \includegraphics[scale=0.4]{MelonInsertion.pdf} \end{array}
\end{equation}
\end{definition}

\begin{definition}
A {\bf melonic graph} is a graph obtained from the graph $G_{\min}$ by iterated melonic moves, in any order.
\end{definition}

An example of such a melonic SYK graph is given in Fig.~\ref{SYK_melon_ex}.
\begin{figure}[ht]
\centering
\includegraphics[scale=0.7]{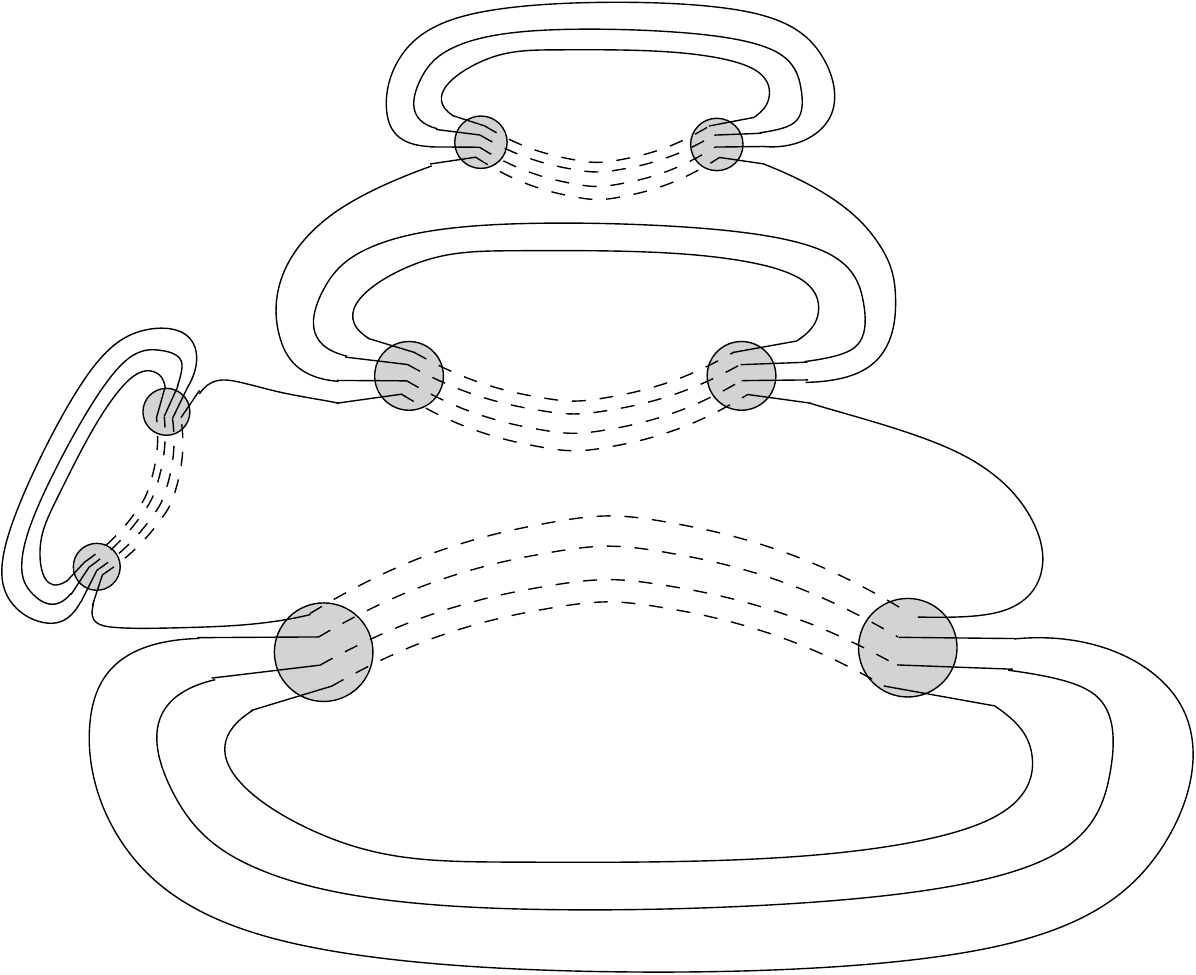}
\caption{An example of melonic graph}
\label{SYK_melon_ex}
\end{figure}

%%%%%%%%%%%%%%
%\section{Proof of the melonic dominance in the large $N$ limit} \label{sec:Proof}
%%%%%%%%%%%%%%

%%%%%%%%%%
\subsubsection{Some properties of melonic graphs} \label{sec:Melons}
%%%%%%%%%%

%The number of faces of melonic graphs is easily found.

\begin{proposition} \label{prop:Melons}
A melonic move adds two vertices and $q-1$ faces to a graph. The number of faces of melonic graphs is
\begin{equation}
F(G) = q + (q-1)\frac{V(G)-2}{2}.
\end{equation}
Thus one has $\delta(G)=1$ for melonic graphs.
\end{proposition}

\begin{proof}
The first statement follows directly from the definition of the melonic move (see Definition \ref{def:melonicmove} above). The number of faces is then obtained by induction. Indeed, $F(G_{\min}) = q$ at $V(G)=2$ for $G_{\min}$, the only melonic graph with two vertices. The induction is completed by using the first statement. The identity $\delta(G)=1$ 
for melonic graphs
follows from the expression of $\delta(G)$ in the definition \eqref{scaling}. %\qed
\end{proof}

Note that, by definition, one can always find a dipole in a melonic graph. 
%It will be convenient to 
Let us also
notice that there is always more than one such dipole.

\begin{proposition} \label{prop:Dipoles}
A melonic graph with at least four vertices has at least two dipoles.
\end{proposition}

\begin{proof}
We proceed by induction on the number of vertices. There is a single melonic graph with four vertices,
\begin{equation}
\begin{array}{c} \includegraphics[scale=.2]{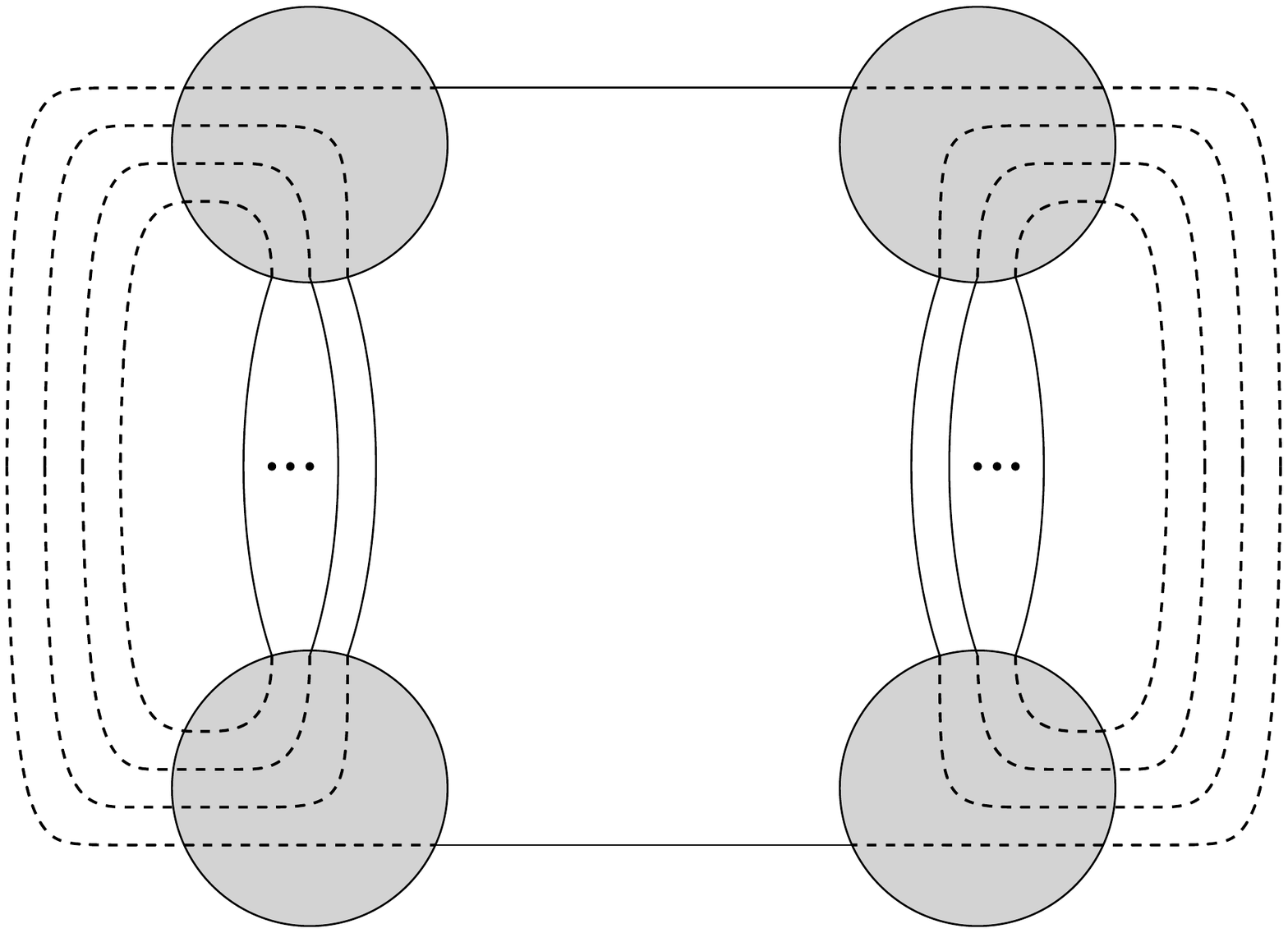} \end{array}
\end{equation}
One can directly check that this graph has indeed has two dipoles.

Assume that the proposition holds for graphs with at most $V-2 \geq 4$ vertices and let $G$ be a melonic graph with $V$ vertices. By construction, 
the graph $G$ can be obtained by a melonic move on the fermionic edge $e$ of the graph $G'$, a melonic graph with $V-2$ vertices. From the induction hypothesis, the graph $G'$ has at least two dipoles. If $e$ is not an edge connecting the two vertices of a dipole, then the melonic move $G'\to G$ increases the number of dipoles. If $e$ connects two vertices of a dipole in the graph $G'$, then the total number of dipoles 
is unchanged between $G'$ and $G$. This comes from the fact that  cutting the edge $e$ destroys one dipole, but the melonic move itself adds one. This concludes the proof.
%\qed
\end{proof}

Melonic graphs satisfy a gluing rule which generalizes the melonic move. Let $G_1, G_2\in\mathbbm{G}$ be two melonic graphs and $e_1$ in $G_1$, $e_2$ in $G_2$ two fermionic lines. If one cuts open $e_1$ in $G_1$ and $e_2$ in $G_2$, then there are two ways to glue the half-edges of $e_1$ with those of $e_2$. To avoid this ambiguity, we use orientations.

\begin{definition}
If $(G, e)$ is a graph $G$ with an oriented fermionic line $e$, denote $G^{(e)}$ the 2-point graph obtained by cutting $e$ into two half-edges with their induced orientations. For two such graphs $(G_1, e_1)$ and $(G_2, e_2)$, denote $G_1^{(e_1)} \star G_2^{(e_2)}$ the unique connected graph obtained by gluing $G_1^{(e_1)}$ with $G_2^{(e_2)}$ in the only way which respects the orientations of the half-edges,
\begin{eqnarray}
G_1^{(e_1)} = \begin{array}{c} \includegraphics[scale=.5]{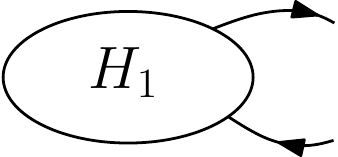} \end{array} \qquad G_2^{(e_2)} = \begin{array}{c} \includegraphics[scale=.5]{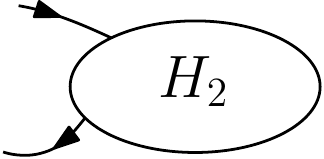} \end{array} 
\nonumber\\ 
\Rightarrow \quad 
G_1^{(e_1)} \star G_2^{(e_2)} = \begin{array}{c} \includegraphics[scale=.5]{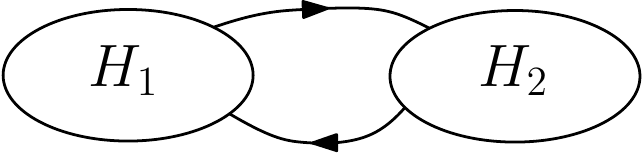} \end{array}
\end{eqnarray}
\end{definition}

\begin{proposition} \label{prop:MelonicGluing}
Let $G_1, G_2\in\mathbbm{G}$ be two melonic graphs and $e_1$ in $G_1$, $e_2$ in $G_2$ two oriented fermionic lines. Then $G_1^{(e_1)} \star G_2^{(e_2)}$ is a melonic graph.
\end{proposition}

\begin{proof}
The result is proved by induction on the number of vertices of the graph $G_1$. If $G_1$ is melonic graph and has two vertices, then $G_1 = G_{\min}$ and the insertion of $G_1^{(e_1)}$ is the melonic move on $e_2$ (for any orientations of $e_1$ and $e_2$).

Assume the proposition holds for graphs $G'_1$ with $V-2$ vertices and consider a new melonic graph $G_1$ with $V$ vertices. It is obtained from a melonic move performed on a fermionic edge $e_1'$ of the melonic graph $G_1'$. 
One then needs to find the edge $e_1$ in $G_1'$, form $G_1^{'(e_1)}\star G_2^{(e_2)}$, which is melonic from the induction hypothesis, and then perform the melonic move on $e_1'$ to get $G_1^{(e_1)}\star G_2^{(e_2)}$, which will thus be a melonic graph also. This is summarized in the following commutative diagram:
\begin{equation}
\begin{array}{c} \includegraphics[scale=.4]{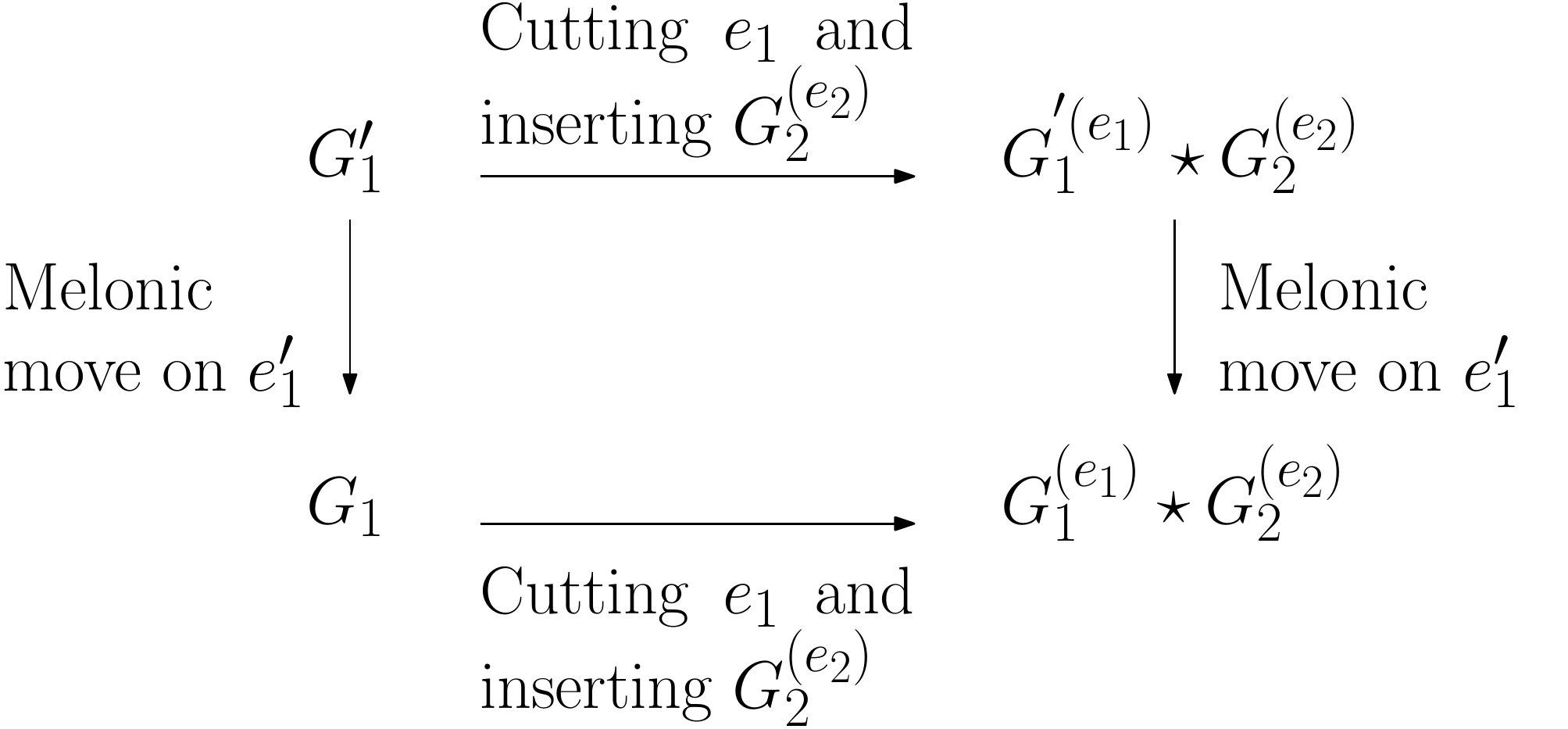} \end{array}
\end{equation}
We thus want to use the path from $G_1'$ to $G_1^{(e_1)}\star G_2^{(e_2)}$ which goes right and then down. When $e_1$ and $e'_1$ are distinct in $G_1'$, this is straightforward:
\begin{equation}
G_1^{'(e_1)} \star G_2^{(e_2)} = \begin{array}{c} \includegraphics[scale=.5]{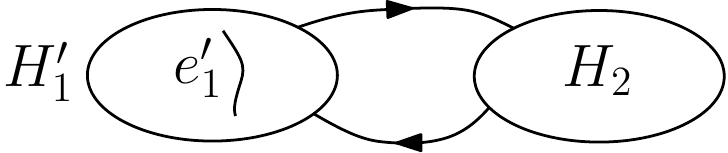} \end{array}
\end{equation}
By the induction hypothesis, this graph is melonic. By definition of the melonic insertion 
%(see again Definition \ref{def:melonicinsertion}, 
the graph remains meloinic 
%Then so it remains 
after the melonic insertion on $e'_1$.
%(by definition).

However, if $e_1$ is incident to or part of the dipole which is inserted from $G_1'$ to $G_1$, it means it does not exist in $G_1'$, as it is created by the melonic move. We distinguish two cases.
\begin{itemize}
\item $e_1$ is a fermionic line connecting the two vertices of the dipole. Then from Proposition \ref{prop:Dipoles} we know that $G_1$ has at least one other dipole. Therefore, one can redefine $G_1'$ has the melonic graph obtained from $G_1$ by removing the latter. Then, the fermionic line $e_1$ can be identified without issues in $G_1'$ and the reasoning above applies.
\item $e_1$ is incident to the dipole, i.e. $(G_1, e_1) = \begin{array}{c} \includegraphics[scale=.4]{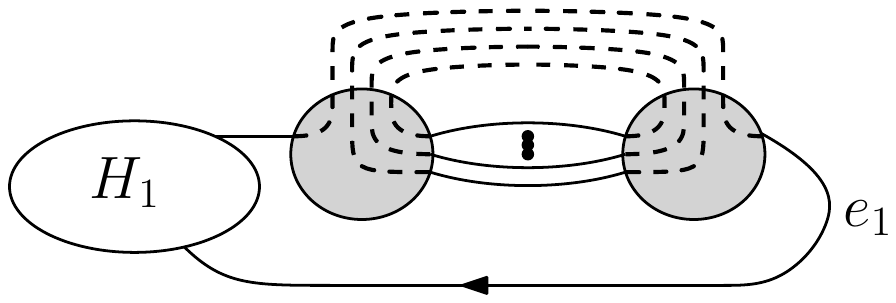} \end{array}$. Then $G^{(e_1)}\star G_2^{(e_2)}$ has the form
\begin{equation}
G^{(e_1)}\star G_2^{(e_2)} = \begin{array}{c} \includegraphics[scale=.4]{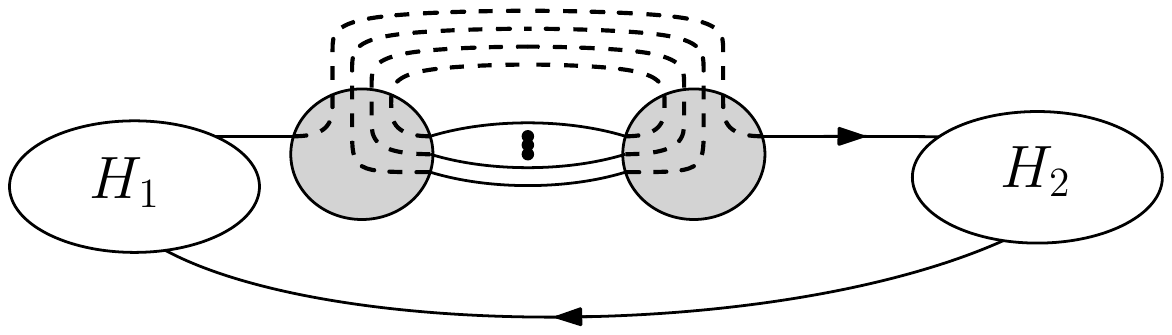} \end{array}
\end{equation}
The graph $G_1'$ is $G'_1 = \begin{array}{c} \includegraphics[scale=.4]{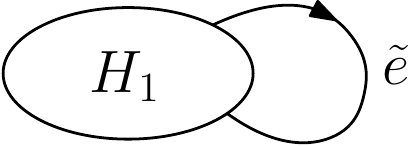} \end{array}$ and is melonic. From the induction hypothesis $G_1^{'(\tilde{e})} \star G_2^{(e_2)} = \begin{array}{c} \includegraphics[scale=.5]{G1StarG2.pdf} \end{array}$ is melonic. Then so is $G^{(e_1)}\star G_2^{(e_2)}$ since it is obtained by a melonic move on $G_1^{'(\tilde{e})} \star G_2^{(e_2)}$.
\end{itemize}
%\qed
\end{proof}

%%%%%%%%%%
\subsubsection{2-cuts} \label{sec:2Cuts}
%%%%%%%%%%

Recall that, following the definition \eqref{scaling} of an SYK degree, we need to identify the graphs which maximize the number of faces at fixed number of vertices. Let us denote the maximal number of faces on $V$ vertices bt 
\begin{equation}
F_{\max}(V) = \max_{\{G\in\mathbbm{G}, V(G) =V\}} F(G)
\end{equation}
and the set of graphs maximizing $F(G)$ at fixed $V$ by
\begin{equation}
\mathbbm{G}_{\max}(V) = \left\{ G\in\mathbbm{G} \text{ s.t. } V(G) = V \quad \text{and} \quad F(G) = F_{\max}(V) \right\}.
\end{equation}

Let us now give the following definition:
\begin{definition} \label{def:2Cut}
A {\bf $2-$cut} is a pair of edges in a graph whose removal (or equivalently cutting) disconnects the graph. 
\end{definition}

Let us now prove that, if there exist two edges in the same face which do not form a 2-cut, the graph is not dominant at large $N$:

\begin{proposition} \label{prop:2Cut}
Let $G\in\mathbbm{G}$ and $e_1, e_2$ two fermionic lines in $G$ which belong to the same face. If $\{e_1, e_2\}$ is not a 2-cut in $G$, then $G\not\in \mathbbm{G}_{\max}(V(G))$.
\end{proposition}

%In other words, 

\begin{proof}
There are two cases to distinguish: whether $\{e_1, e_2\}$ is a 2-cut or not in $G_0$.
\begin{enumerate}
%[style=wide]
\item $\{e_1, e_2\}$ is not a 2-cut in $G_0$.\\
We draw $G$ as
\begin{equation}
G = \begin{array}{c} \includegraphics[scale=.7]{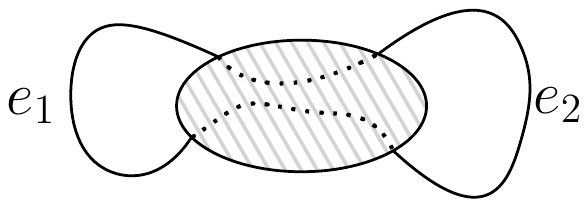} \end{array}
\end{equation}
where the dotted line represents the paths alternating fermionic lines and strands of disorder lines which constitute the face of $e_1$ and $e_2$.

Now consider $G'$ obtained by cutting $e_1$ and $e_2$ and regluing the half-lines in the unique way which creates one additional face,
\begin{equation}
G' = \begin{array}{c} \includegraphics[scale=.7]{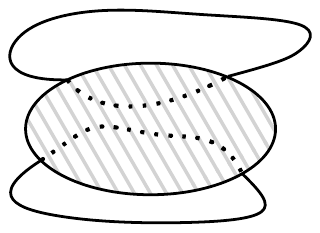} \end{array}
\end{equation}
$G_0'$ is connected since $\{e_1, e_2\}$ is not a 2-cut in $G_0$, and hence $G'\in\mathbbm{G}$. No other faces of $G$ are affected. Therefore $F(G') = F(G)+1$ and thus $G\not\in \mathbbm{G}_{\max}(V(G))$.

\item $\{e_1, e_2\}$ is a 2-cut in $G_0$.\\
An example of this situation is when $G_0$ is melonic but $G$ is not because the disorder lines are added in a way which does not respect melonicity.

In this case, $G$ looks like
\begin{equation}
G = \begin{array}{c} \includegraphics[scale=.7]{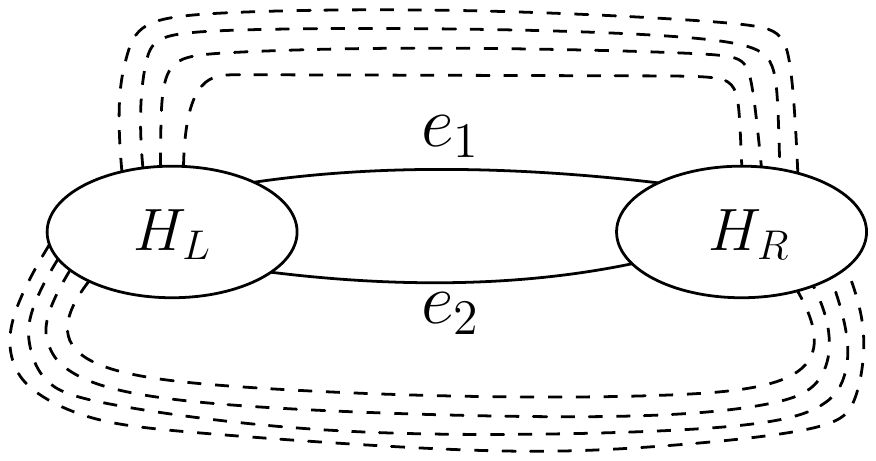} \end{array}
\end{equation}
i.e. $H_L$ and $H_R$ are both connected, and the only lines between them are $e_1, e_2$ and some disorder lines. Consider $G'$ obtained by cutting $e_1$ and $e_2$ and regluing the half-lines as follows
\begin{equation}
G' = \begin{array}{c} \includegraphics[scale=.7]{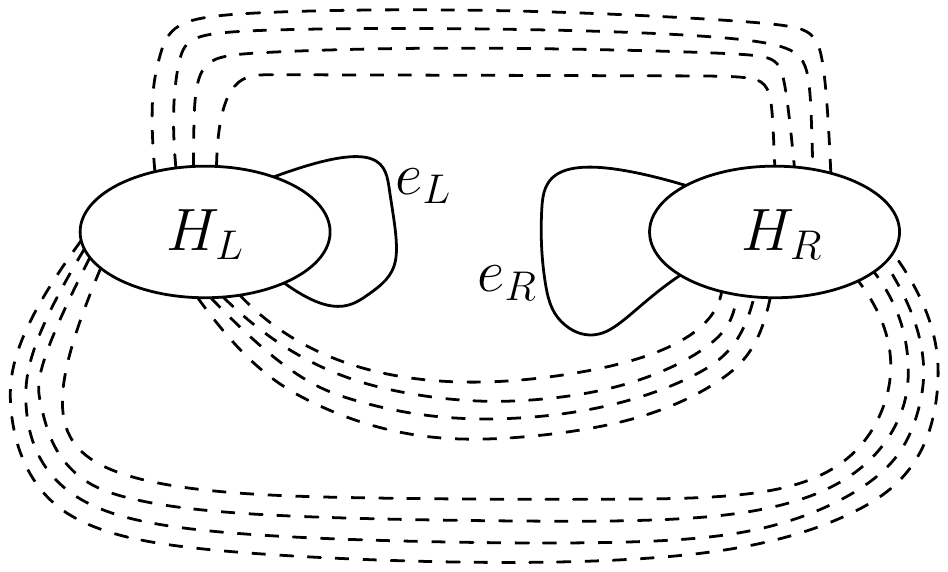} \end{array}
\end{equation}
Notice that $G'\not\in\mathbbm{G}$ since $G'_0$ consists of two connected components $G_{0L}'$ and $G_{0R}'$.

Consider a disorder line $e_0$ between them. It joins two vertices $v_L$ in $G_{0L}'$ and $v_R$ in $G_{0R}'$. We perform the contraction of the disorder line $e_0$ as follows
\begin{equation}
\begin{aligned}
&G' = \begin{array}{c} \includegraphics[scale=.5]{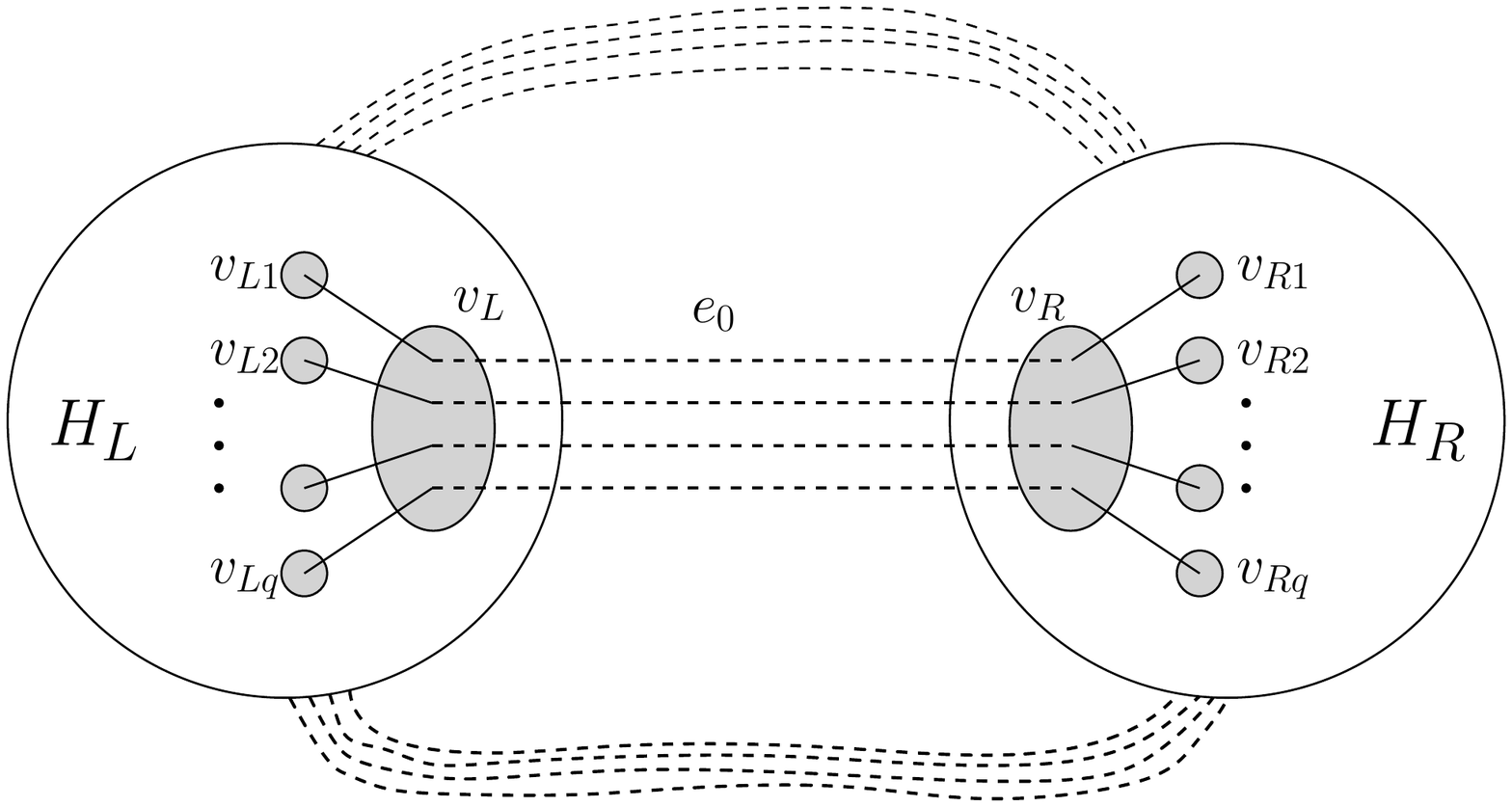} \end{array}\\
\to\quad &G'' = \begin{array}{c} \includegraphics[scale=.5]{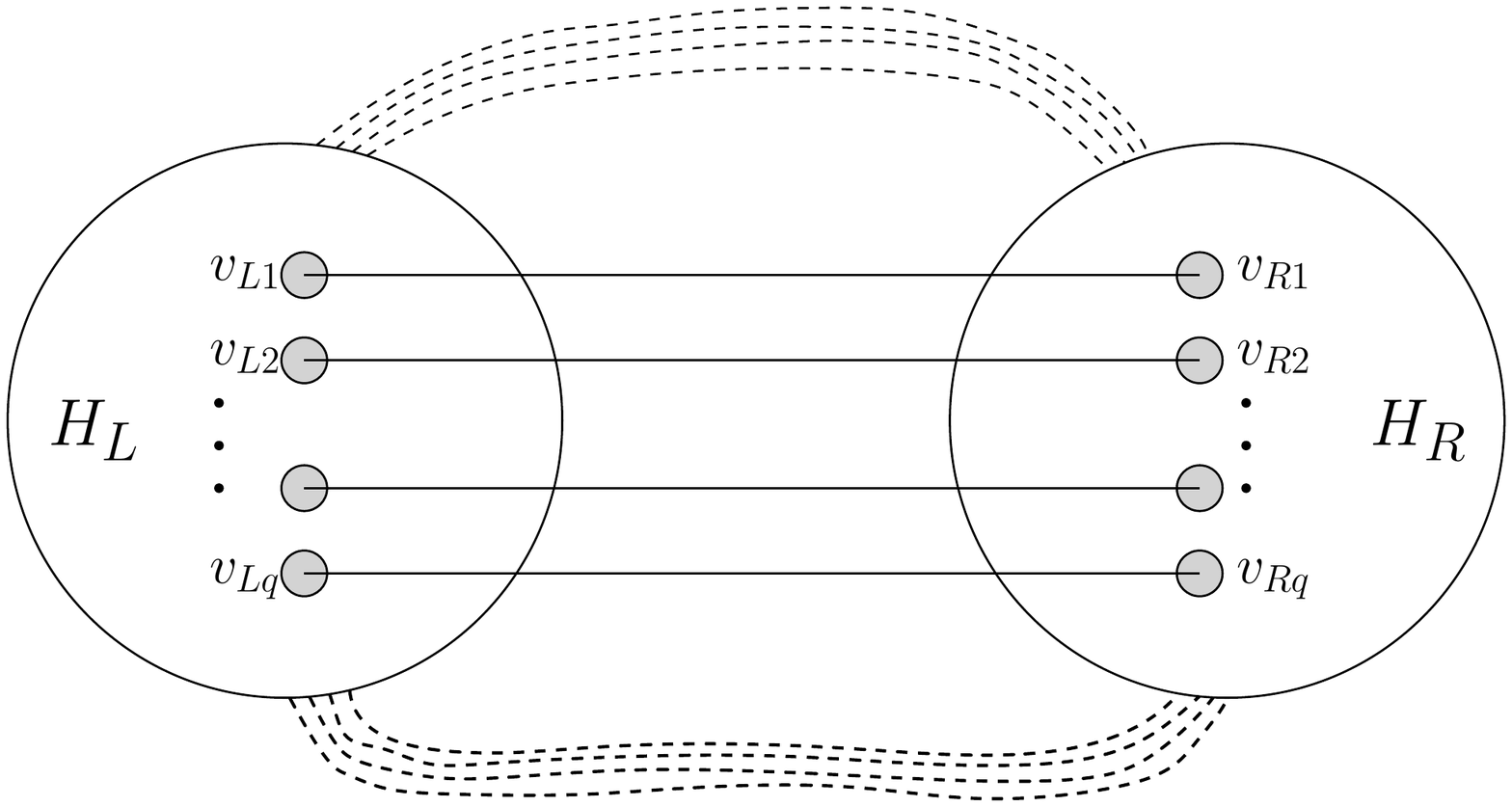} \end{array}
\end{aligned}
\end{equation}
It removes $v_L, v_R$ and $e_0$ and joins the pending fermionic lines which were connected by the strands of $e_0$. The key point is that $G''\in\mathbbm{G}$ now since the contraction of $e_0$ connects the two disjoint components of $G_0'$ by $q$ fermionic lines. 

Let us now analyze the variations of the number of faces from $G$ to $G''$. First from $G$ to $G'$: in $G$ the lines $e_1, e_2$ belong to the same face, while $e_L$ and $e_R$ may or may not belong to the same face in $G'$, hence 
\begin{equation}
F(G) \leq F(G').
\end{equation}
Then the contraction of $e_0$ does not change the number of faces. Indeed, the faces of $G'$ which do not go along $e_0$ are not affected. As for those which go along $e_0$, they follow paths
\begin{equation}
v_{Li}\to v_L\to v_R \to v_{Ri}
\end{equation}
for $i=1, \dotsc, q$ (some of those $q$ paths may belong to common faces). In $G''$, they become paths going directly from $v_{Li}$ to $v_{Ri}$. There is thus a 1-to-1 correspondence between the faces of $G'$ and those of $G''$. Therefore, $F(G) \leq F(G'')$.

%Here for this to be true, it is key that that $e_0$ connects two disjoint components of $G_0'$. 

To conclude the proof, notice that $G''$ has two vertices less than $G$. Therefore we can perform a melonic insertion on any fermionic line of $G''$ to get a graph $\tilde{G}\in\mathbbm{G}$ with $V(G) = V(\tilde{G})$ and %crucially 
\begin{equation}
F(\tilde{G}) = F(G'') + q-1
\end{equation}
as in Proposition \ref{prop:Melons}. For $q>1$ it comes that $F(G)< F(\tilde{G})$ and thus $G\not\in\mathbbm{G}_{\max}(V(G))$.
%\end{description}
\end{enumerate}
%\qed
\end{proof}

%Using the notation of Proposition \ref{prop:2Cut}, Theorem \ref{thm} rewrites as
Let us now prove that
\begin{equation}
\bigcup_{\text{$V$ even}} \mathbbm{G}_{\max}(V) = \left\{ G\in\mathbbm{G} \text{ s.t. } \delta(G)=1 \right\} = \left\{ \text{Melonic graphs}\right\}.
\end{equation}

%%%%%%%%%%
%\subsection{Large $N$ limit} \label{sec:Thm}
%%%%%%%%%%

%We now prove Theorem \ref{thm}.

This is equivalent to 
the main theorem of this subsection, which states:

\begin{theorem} \label{thm}
The weight of $G\in\mathbbm{G}$ is bounded by:
\begin{equation}
\delta(G) \leq 1
\end{equation}
Moreover, the graphs such that $\delta(G) = 1$ are the melonic graphs.
\end{theorem}

\begin{proof}
We proceed by induction.
The graph $G_{\min}$ is melonic by definition. It has $F(G_{\min})=q$ and $V(G_{\min})=2$ hence satisfies $\delta(G_{\min}) = 1$. Since it is the only graph on two vertices, the theorem indeed holds on two vertices.

Let $V\geq 4$ even. We assume the theorem is true up to $V-2$ vertices and consider $G\in\mathbbm{G}$ with $V(G) = V$ vertices. 

We need to investigate pairs $\{e_1, e_2\}$ with $e_1, e_2$ two fermionic lines belonging in a common face. Notice that such a pair exists. If it was not the case, then all faces would be of length 2 (i.e. one fermionic line and one disorder line) which implies $G=G_{\min}$, which is impossible since $G$ has $V\geq 4$ vertices.

Let $\{e_1, e_2\}$ be a pair of fermionic edges belonging to  the same face. Due to Proposition \ref{prop:2Cut}, we know that it is a 2-cut in $G$. The graph therefore takes the form
\begin{equation}
G = \begin{array}{c} \includegraphics[scale=.6]{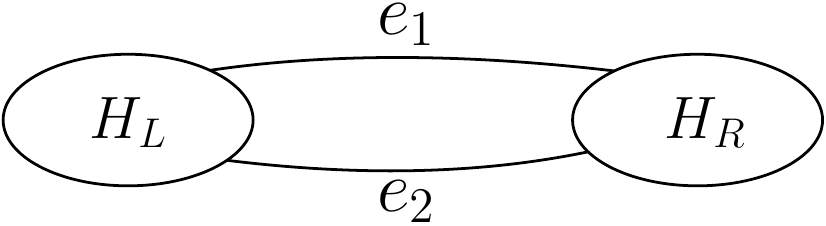} \end{array}
\end{equation}
where $H_L, H_R$ are connected, 2-point graphs (in the sense that $e_1$ and $e_2$ are hanging out). We cut $e_1$ and $e_2$ and glue the resulting half-lines to close $H_L$ and $H_R$ into $G_L, G_R$, and use ``reverse'' orientations as follows
\begin{equation}
G_L = \begin{array}{c} \includegraphics[scale=.6]{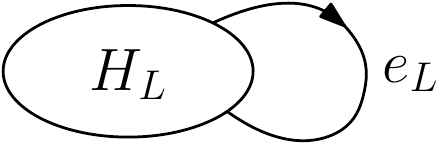} \end{array} \qquad G_R = \begin{array}{c} \includegraphics[scale=.6]{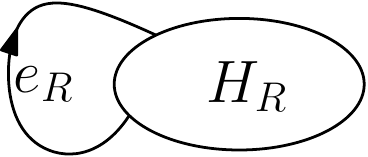} \end{array}
\end{equation}
%In the notation of Proposition \ref{prop:MelonicGluing}, 
We thus have
\begin{equation}
G = G_L^{(e_L)} \star G_R^{(e_R)}.
\end{equation}

Since the edges $e_1$ and $e_2$ belong to the same face, we have 
\begin{equation}
F(G) = F(G_L) + F(G_R) - 1,
\end{equation}
and $F(G)$ is maximal iff $F(G_L)$ and $F(G_R)$ are. From the induction hypothesis, this requires $G_L$ and $G_R$ to be melonic. Then $G = G_L^{(e_L)} \star G_R^{(e_R)}$ is melonic too according to Proposition \ref{prop:MelonicGluing}. %\qed
\end{proof}

%We get as a corollary of Theorem \ref{thm} and Proposition \ref{prop:2Cut} a characterization of melonic graphs.

Let us end this subsection with the following result:

\begin{corollary}
A graph $G\in\mathbbm{G}$ is melonic iff all pairs $\{e_1, e_2\}$ of fermionic lines which belong in a common face are 2-cuts.
\end{corollary}

%section{The Gross-Rosenhaus model and the colored SYK model}

\medskip

The large $N$ melonic dominance of the SYK model being now proved, in the rest of this review
we represent SYK graphs with disorder edges as regular edges, without the stranded structure explained in this section.

\section{The colored Sachdev-Ye-Kitaev model}
\label{sec:coloredSYK}

\subsection{Definition of the real and complex model}
\label{sec:defcolored}

The SYK generalization we study in this section contains $q$ flavors of fermions. Moreover, each fermion of a given flavor appears exactly once in the interaction and the Lagrangian couples $q$ fermions together. The action writes:
\begin{equation}
\label{act:noi}
S=\int d\tau \left( 
\frac 12 \sum_{f=1}^q\sum_{i=1}^N \psi_i^f \frac{d} {dt} \psi_i^f
- \frac{i^{q/2}}{q!}\sum_{i_1,\ldots, i_q=1}^N
j_{i_1\ldots i_q}\psi_{i_1}^{1}\ldots \psi_{i_q}^{q},
\right)
\end{equation}
Note that we use superscripts to denote the flavor. Moreover,  in order to simplify the notations,  the model has $q\cdot N$ fermions - we have $N$ fermions of a given flavor.

\medskip

The SYK generalization introduced above is a particular case of the Gross-Rosenhaus generalization  \cite{Gross}. Indeed, Gross and Rosenhaus took a number $f$ of flavors, with $N_a$ fermions of flavor $a$, each appearing $q_a$ times in the interaction, such that $N=\sum_{a=1}^f N_a$ and $q=\sum_{a=1}^f q_a$. In our case, the number $f$ of flavors is equal to $q$, $q_a=1$ and $N_a=N$ (recall that we have now a total of $q\cdot N$ fermions).
The action 
\eqref{act:noi}
is close in spirit to the action of the colored tensor model \cite{colored}, hence the name colored SYK model.

\medskip

Thus, the Feynman graphs obtained through perturbative expansion of the action \eqref{act:noi} are edge-colored graphs where the colors are the flavors. At each vertex, each of the $q$ fermionic fields which interact has one of the $q$ flavors, and each flavor is present exactly once.

An example of such a Feynman graph is given on the right of Fig. \ref{melonSYK} (while on the left side we have a melonic graph of the SYK model). 
%Once again, we have represented the disorder by a dashed line. 
The disorder edge, represented, as in the previous section, as a dashed edge,  can be considered to have the fictitious flavor $0$.

\begin{figure}
\begin{center}
\includegraphics[scale=0.9]{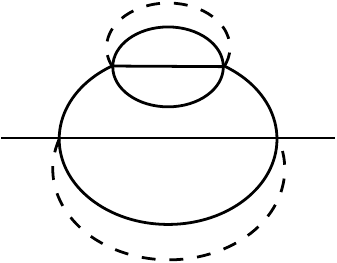}\hspace{2cm} \includegraphics[scale=0.9]{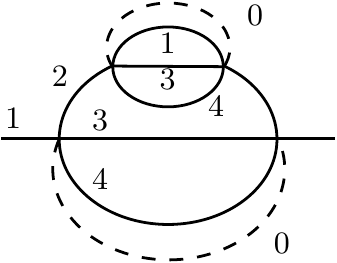}
\caption{\label{melonSYK}Melonic graphs of the SYK and colored SYK models}
\end{center}
\end{figure}

\medskip

%In this paper we will focus on the real model \eqref{act:noi}. 
There is also a complex version of the model \eqref{act:noi}, version initially mentioned in 
\cite{gurau-ultim}. This latter version can be easily obtained by 
%considering the propagation 
%from $\psi$ to $\bar \psi$ and by
using complex fields and by 
 considering the interacting term in \eqref{act:noi} as well as its complex conjugate:
\begin{equation}
\label{act:noir}
\int d\tau \left( 
\frac 12 \sum_{f=1}^q\sum_{i=1}^N \bar \psi_i^f \frac{d} {dt} \psi_i^f
- \frac{i^{q/2}}{q!}\sum_{i_1,\ldots, i_q=1}^N
j_{i_1\ldots i_q}\psi_{i_1}^{1}\ldots \psi_{i_q}^{q}
- \frac{(-i)^{q/2}}{q!}\sum_{i_1,\ldots, i_q=1}^N
\bar j_{i_1\ldots i_q}\bar \psi_{i_1}^{1}\ldots \bar\psi_{i_q}^{q},
\right)
\end{equation}
The Feynman graphs obtained through perturbative expansion of the complex action have the same structure as the one explained above for the real model \eqref{act:noi}. However, in the complex case, one has two types of vertices, which we can refer to as white and black, as it is done in the tensor model literature 
(see, for example,  the book \cite{gurau-book} and references within).  
Each edge connects a white to a black vertex. The Feynman graphs of \eqref{act:noir} are thus the subset of the Feynman graphs of \eqref{act:noi} which are bipartite. This is a feature which  simplifies the diagrammatic analysis of the complex model. 
%In this paper however, we will have to deal with typically non-bipartite graphs.

\subsection{Diagrammatics of the real model} \label{sec:NLO}
%%%%%%%%%%%%%

%At fixed couplings $j_{i_1 \dotsc i_q}$, the Feynman graphs are $q$-regular edge-colored graphs: graphs with a color from $\{1, \dotsc, q\}$ on each edge and such that all colors are incident exactly once one each vertex.

%In order to study the $1/N$ expansion, one must average over the disorder with the covariance
%\begin{equation}
%\langle j_{i_1 \dotsc i_q} j_{l_1 \dotsc l_q} \rangle \sim \frac1{N^{q-1}} \prod_{k=1}^q \delta_{i_k, l_k}
%\end{equation}
%and similarly for $\langle j_{i_1 \dotsc i_q} \bar{j}_{l_1 \dotsc l_q} \rangle$ in the complex case. Each graph is thus turned into a sum over Wick pairings which can be represented with edges carrying a new color, say the color $0$.

%This adds to the fermionic Feynman graphs an additional set of edges with color 0. An edge of color 0 must join two vertices. 

This subsection follows the original article \cite{Bonzom:2017pqs}.

For each color $i\in\{1, \dotsc, q\}$, a Feynman graph has cycles (i.e. closed paths) which alternate the colors $0$ and $i$. We call them \textbf{faces of colors $0i$}. This terminology is an extension of matrix models where those cycles are faces of ribbon graphs.

We denote by $F_{0i}(G)$ the number of faces of colors $0i$ for $i=1, \dotsc, q$ of a graph $G$, 
\begin{equation}
F_0(G) = \sum_{i=1}^q F_{0i}(G)
\end{equation}
the total number of faces which have the color $0$. We further denote by $E_0(G)$ the number of edges of color 
$0$ of the graph $G$ (which is, as in the SYK case of the previous section, half the number of vertices of the graph $G$). 
%It is easy to see that $G$ has a free sum for each face of colors $0i$ which sums up to $N$. It thus receives a total weight
In the large $N$ limit, the Feynman amplitude of a colored SYK graphs is given by 
$$N^{\chi_0(G)}$$
where the {\textbf{colored SYK degree}} is:
\begin{equation}
\label{SYKdeg}
 \chi_0(G)= F_0(G) - (q-1)E_0(G).
\end{equation}
%Gurau's theorem on the $1/N$ expansion of tensor models (see, for example, the version presented in \cite{uncolored}) ensures that this is bounded,
Using colored tensor model results (see, for example \cite{uncolored}) one can prove:
\begin{equation}
\chi_0(G) = F_0(G) - (q-1)E_0(G) \leq \begin{cases} 1 & \text{if $G$ is a vacuum graph,}\\
0 & \text{if $G$ is a 2-point graph.} \end{cases}
\end{equation}
The case of 4-point graphs will be discussed later. 

In the language of \cite{uncolored}, the graphs of the colored SYK models have a single bubble, i.e. a single connected component after removing the edges of color $0$, since this bubble is the underlying, connected fermionic graph at fixed couplings.

\subsubsection{LO, NLO of vacuum and 2-point graphs} \label{sec:2Pt}
%%%%%%%%%%%%

Notice that all 2-point graphs are obtained by cutting an edge $e$ of color $i\in\{1, \dotsc, q\}$ in a vacuum graph $G$. Since there is a single face, with colors $0i$, which goes through $e$ in $G$, cutting it decreases the exponent of $N$ by one exactly.

To study $\chi_0(G)$, we perform in $G$ the contraction of the edges of color 0 to get the graph $G_{/0}$,
\begin{equation}
\label{Contraction}
\begin{array}{c} \includegraphics[scale=.6]{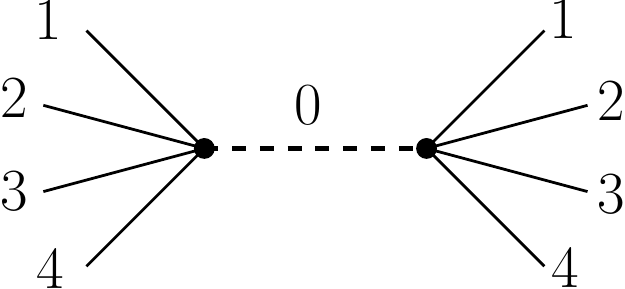} \end{array} \qquad \underset{/0}{\to} \qquad \begin{array}{c} \includegraphics[scale=.6]{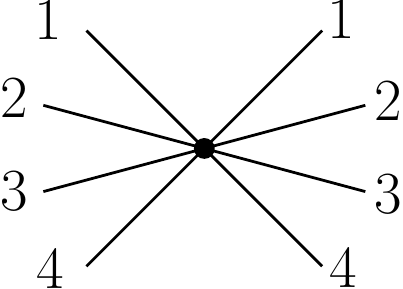} \end{array}
\end{equation}
This means that two vertices of $G$ connected by an edge of color 0 become a single vertex in $G_{/0}$. The map $G\mapsto G_{/0}$ is not one-to-one because of this. Nevertheless, in the complex case, where $G$ is bipartite, it can be made one-to-one by orienting the edges from, say, $\psi^i$ to $\bar{\psi}^i$, i.e. from white to black vertices. Then the edges of $G_{/0}$ are oriented and this is sufficient to reconstruct $G$. In the real case, $G$ is not always bipartite and there are typically several graphs $G$ for the same $G_{/0}$.

\emph{The main property of $G_{/0}$ is that all $q$ colors are incident exactly twice on each vertex}. Therefore, the edges of color $i$ form a disjoint set of cycles (we recall that a cycle is a closed path which visits its vertices only once). Let $\ell_i(G_{/0})$ be the number of cycles of edges of color $i$. From the construction of $G_{/0}$, its cycles of color $i$ are the faces of colors $0i$ of $G$,
\begin{equation}
F_{0i}(G) = \ell_i(G_{/0}).
\end{equation}
Let us introduce $L(G_{/0})$ the cyclomatic number of $G_{/0}$, i.e. its number of independent cycles, or first Betti number. As is well known, it is the number of edges of $G_{/0}$ minus its number of vertices plus one. The number of edges of $G_{/0}$ is the number of edges of $G$ with colors in $\{1, \dotsc, q\}$, thus $q E_0(G)$. The number of vertices of $G_{/0}$ simply is $E_0(G)$, so that
\begin{equation}
L(G_{/0}) = (q-1) E_0(G) + 1.
\end{equation}
This shows that 
\begin{equation}
\chi_0(G) = \sum_{i=1}^d \ell_i(G_{/0}) - L(G_{/0}) + 1
\end{equation}
which has a simple graphical interpretation: it is minus the number of multicolored cycles. Indeed, a cycle can be single-colored or multicolored. The former are counted by $\sum_{i=1}^q \ell_i(G_{/0})$ while $L(G_{/0})$ counts the total number of cycles. Therefore their difference leaves precisely the number of cycles $\ell_m(G_{/0})$ which are multi-colored, up to a sign,
\begin{equation} \label{DegreeCycles}
\chi_0(G) = - \ell_m(G_{/0}) + 1.
\end{equation}
The classification of graphs $G$ with respect to $\chi_0(G)$ is therefore obtained from $\ell_m(G_{/0})$.

%\subsubsubsection{Leading order}
%%%%%%%%%%%%

The LO large $N$ limit graphs are graphs which satistfy $\ell_m(G_{/0}) = 0$, i.e. $G_{/0}$ has no multicolored cycles. It means that it is made of single-colored cycles which are glued without forming additional cycles. The corresponding graphs $G$ are easily seen to be melonic. Indeed, one starts from $G_{/0}$ being a simple single-colored cycle of color $i\in\{1, \dotsc, q\}$ with loops of all other colors on its vertices. Then each vertex of $G_{/0}$ is replaced with a pair of vertices and each loop becomes an edge between them. The color 0 from the average over disorder is added between the vertices of each pair too. One gets a melonic cycle as follows,
\begin{equation}
G_{/0} = \begin{array}{c} \includegraphics[scale=.6]{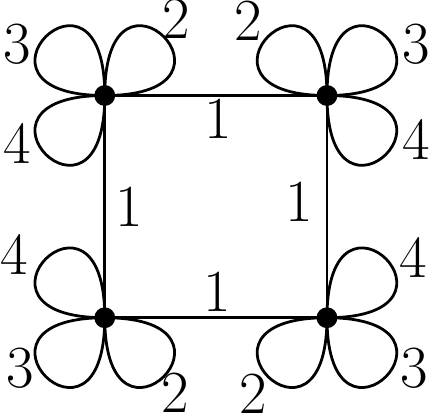} \end{array} \qquad \Rightarrow \qquad G = \begin{array}{c} \includegraphics[scale=.6]{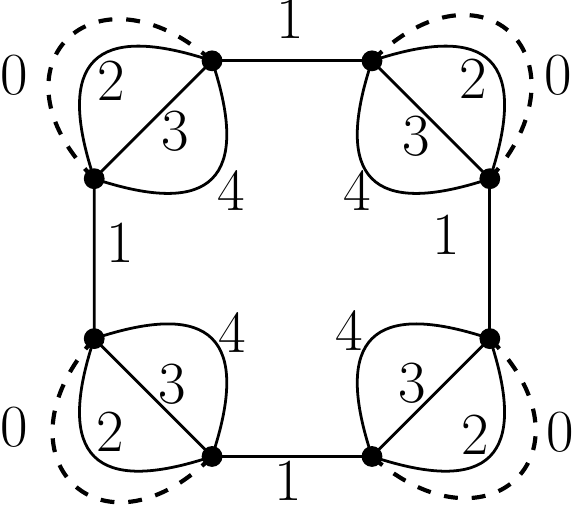} \end{array}.
\end{equation}
More general $G_{/0}$ are obtained by cutting a loop, say of color 2, and replacing it with a cycle and loops attached to its vertices. This corresponds to cutting an edge of color 2 in $G$ and gluing another melonic cycle. This recursive process generates all the graphs corresponding to the large $N$ limit.

The large $N$ $2-$point function is simply obtained by cutting an edge of color $i\in\{1, \dotsc, q\}$. From the above recursive process, one finds the following description of the large $N$, fully dressed propagator
\begin{equation} \label{Melonic2Pt}
\begin{array}{c} \includegraphics[scale=.6]{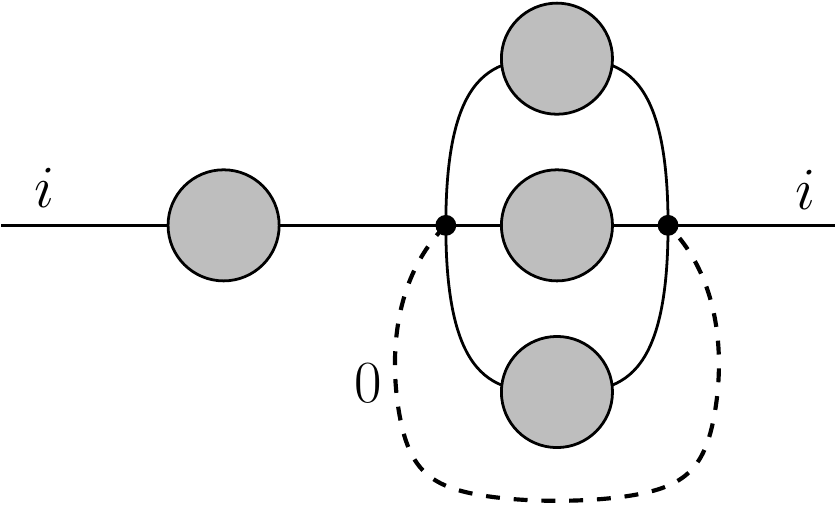} \end{array}
\end{equation}
where each gray blob reproduces the same structure.

%%%%%%%%%%%
%\subsubsection{Next-to-leading order and chains}
%%%%%%%%%%%

One can check that replacing an edge in $G$ with any LO 2-point function of the form \eqref{Melonic2Pt} does not change $\chi_0(G)$. Therefore, \emph{all solid edges in the remaining of the article are large $N$, fully dressed propagators}.

2-point functions in the representation as $G_{/0}$ are simply obtained by contracting all edges of color 0 of 2-point graphs $G$. Therefore, solid edges in $G_{/0}$ will also represent fully dressed propagators from now on.

At NLO, one finds graphs such that $\ell_m(G_{/0}) = 1$, i.e. $G_{/0}$ has a single multicolored cycle. Compared to the large $N$ limit, this means that one obtains $G_{/0}$ by gluing single-colored cycles (with loops attached to their vertices) so as to form a single multicolored cycle.

Considering that solid edges are fully dressed 2-point functions, the NLO graphs $G_{/0}$ are completely characterized by the length $n$ of the multicolored cycle with colors $i_1, i_2, \dotsc, i_n$. For instance at length $n=6$:
\begin{equation}
G_{/0}^{\text{NLO}} = \begin{array}{c} \includegraphics[scale=.6]{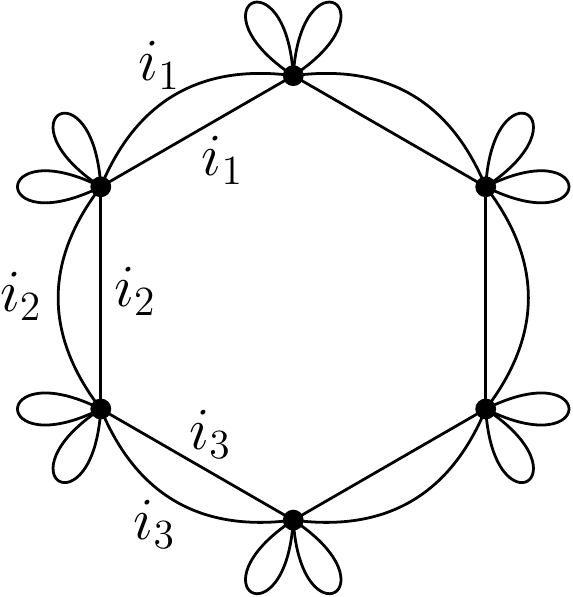} \end{array}
\end{equation}

To find the corresponding graphs $G$, one splits each vertex of $G_{/0}$ into two vertices connected by an edge of color 0 and so that each color is incident exactly once on each vertex. There are several ways to connect the edges of color $i_j$ and $i_{j+1}$ to a pair of vertex. Overall, this leads to the two following families of graphs,
\begin{equation} \label{VacuumNLO}
G^{\text{NLO}} = \begin{array}{c} \includegraphics[scale=.5]{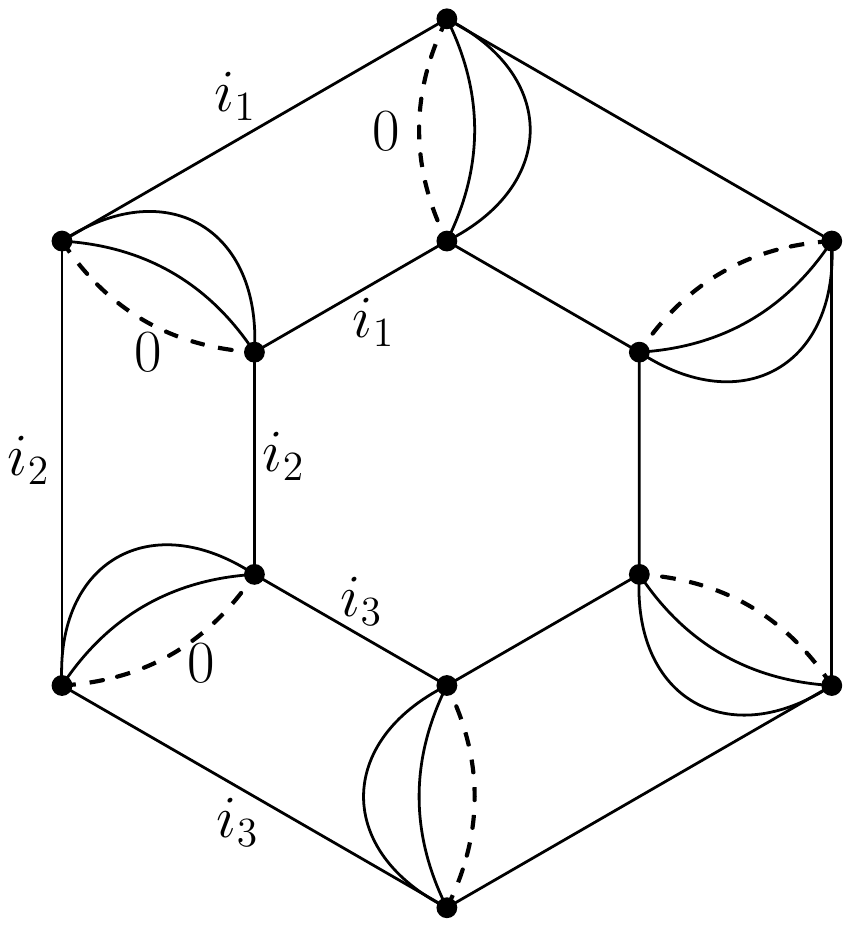} \end{array} \qquad \text{and} \qquad \tilde{G}^{\text{NLO}} = \begin{array}{c} \includegraphics[scale=.5]{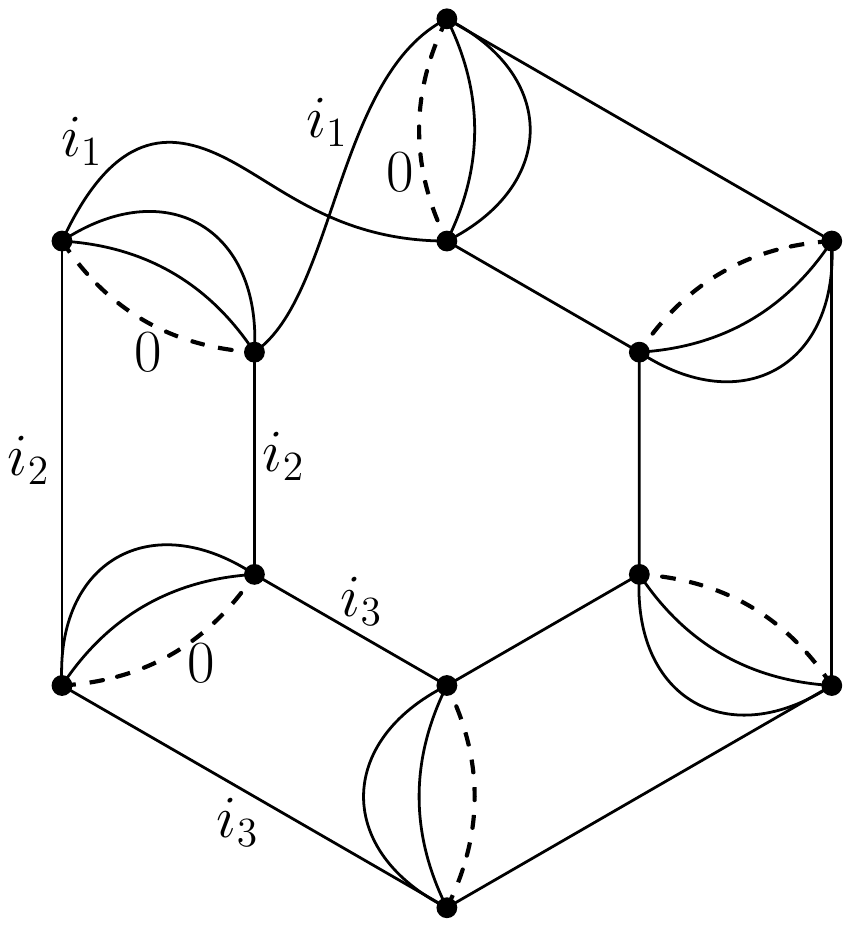} \end{array}
\end{equation}
Notice that $G^{\text{NLO}}$ is bipartite (recall that the melonic 2-point functions are) while $\tilde{G}^{\text{NLO}}$ is not.
%\footnote{In a complex theory, only the bipartite graphs would be admissible.}. 
The graph $\tilde{G}^{\text{NLO}}$ is obtained from $G^{\text{NLO}}$ by crossing two edges, say with colors $i_1$. Adding more crossings is always equivalent to $G^{\text{NLO}}$ (for an even number of crossings) or $\tilde{G}^{\text{NLO}}$ (for an odd number of crossings).

To remember that the graphs above can have arbitrary lengths, and also to offer a convenient representation of NLO 2-point functions (to come below), we introduce \emph{chains} which are 4-point graphs,
\begin{equation} \label{SYKChain}
\begin{array}{c} \includegraphics[scale=.6]{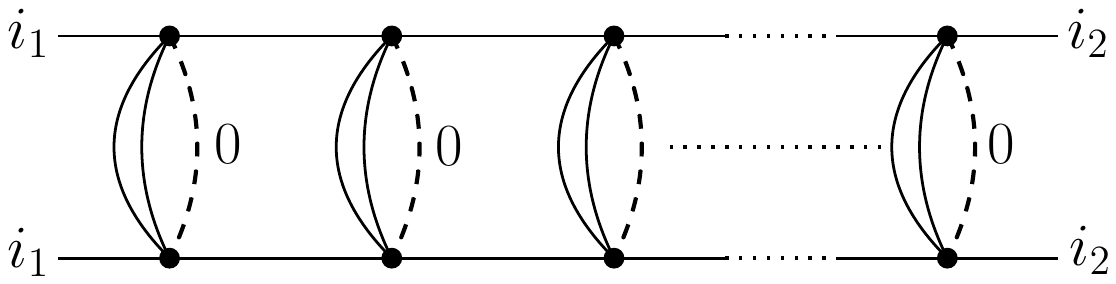} \end{array}
\end{equation}
A combinatorial detail of importance is that a chain can have down to two vertices only, and has at least two vertices unless stated otherwise. % Notice also that the external legs are amputated: they do not carry any 2-point insertions.
We will represent arbitrary choices of chains as boxes,
\begin{equation} \label{SYKChainVertex}
\begin{array}{c} \includegraphics[scale=.6]{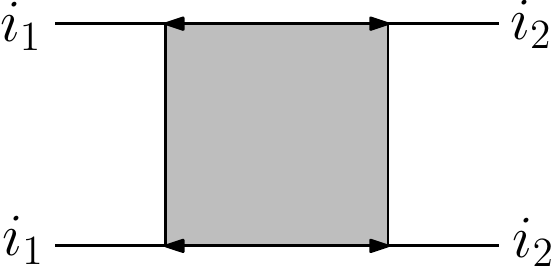} \end{array}
\end{equation}
where the arrows indicate the direction of the chain (the box by itself being symmetric).

This enables to represent the two families of NLO vacuum graphs as
\begin{equation}
\label{Vacuum_NLO_Schemes}
G^{\text{NLO}} = \begin{array}{c} \includegraphics[scale=.5]{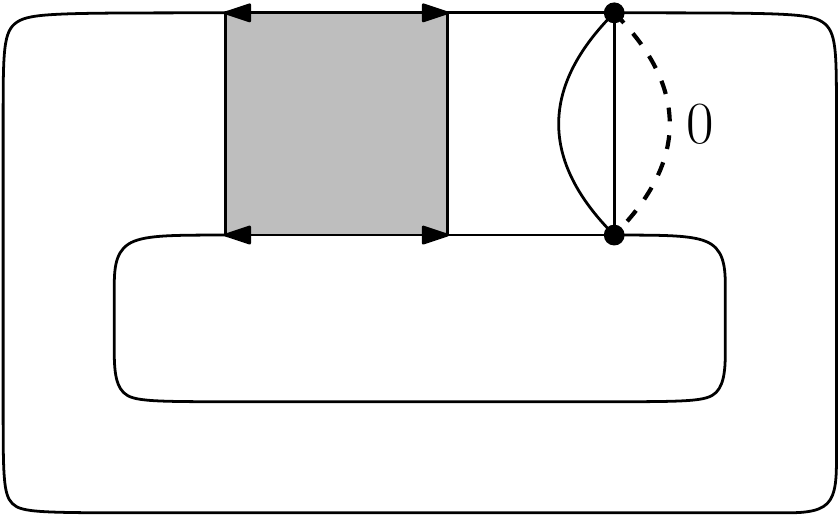} \end{array} \qquad \text{and} \qquad \tilde{G}^{\text{NLO}} = \begin{array}{c} \includegraphics[scale=.5]{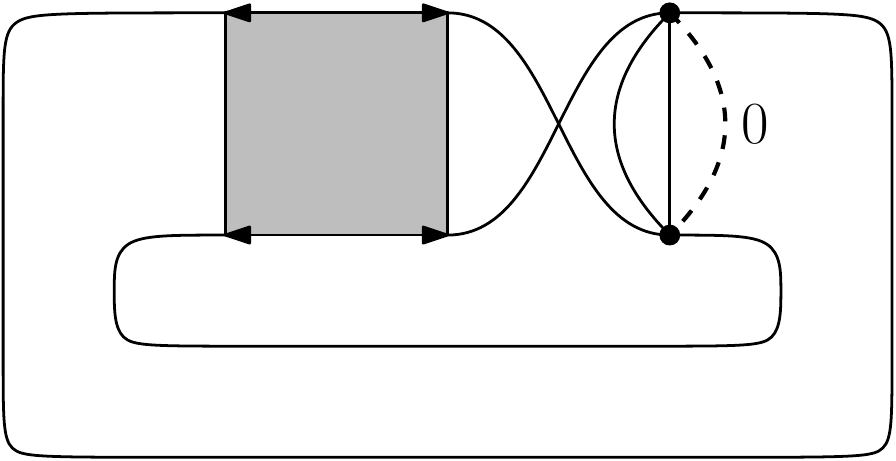} \end{array}
\end{equation}

To get the $2-$point NLO graphs from vacuum graphs, it is sufficient to cut an edge of a given color $i\in\{1, \dotsc, q\}$ in a NLO vacuum graph. However, we have to remember that we have used dressed propagators in \eqref{VacuumNLO}. For instance, $G^{\text{NLO}}$ really is
\begin{equation} \label{NLOGraph}
\begin{array}{c} \includegraphics[scale=.6]{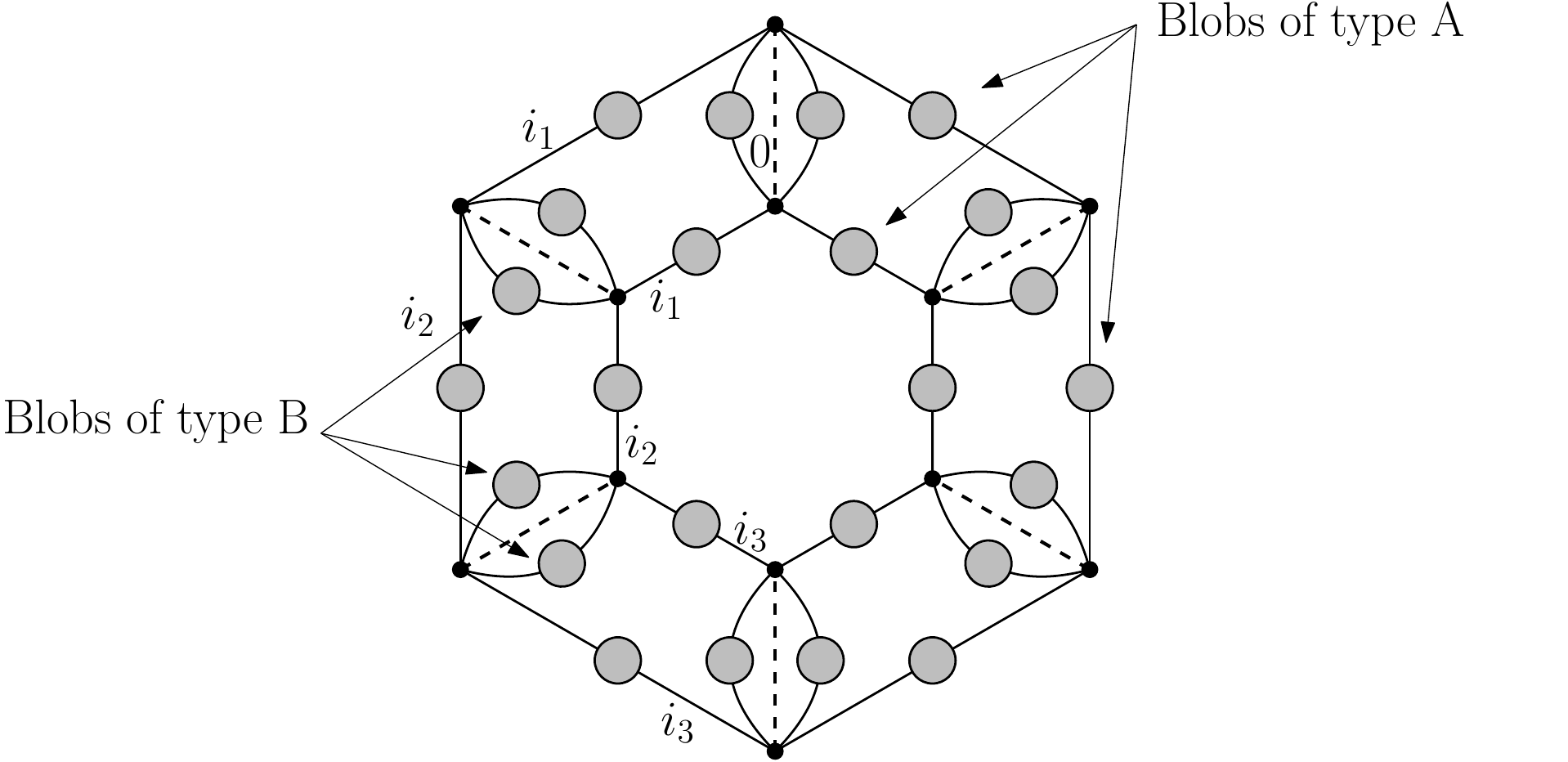} \end{array}
\end{equation}
where the gray blobs represent arbitrary, LO 2-point functions. One might (not necessarily but typically) cut an edge which is contained in a gray blob of \eqref{NLOGraph}. There are two cases to distinguish depending on where an edge is cut in \eqref{NLOGraph}, because there are two types of blobs in \eqref{NLOGraph}.
\begin{itemize}
\item Blobs of type A are inserted on the $2n$ edges of colors $i_1, \dotsc, i_n$ which are characterized as follows: such an edge connects two vertices which are not incident to the same edge of color 0.
\item Blobs of type B are the others: they are inserted on the edges whose end-points are incident to the same edge of color 0.
\end{itemize}

If the cut edge is chosen within a blob of type A, then there are two types of NLO 2-point graphs:
\begin{equation} \label{NLO2Pt1}
G_2^{\text{NLO}(1)} = \begin{array}{c} \includegraphics[scale=.45]{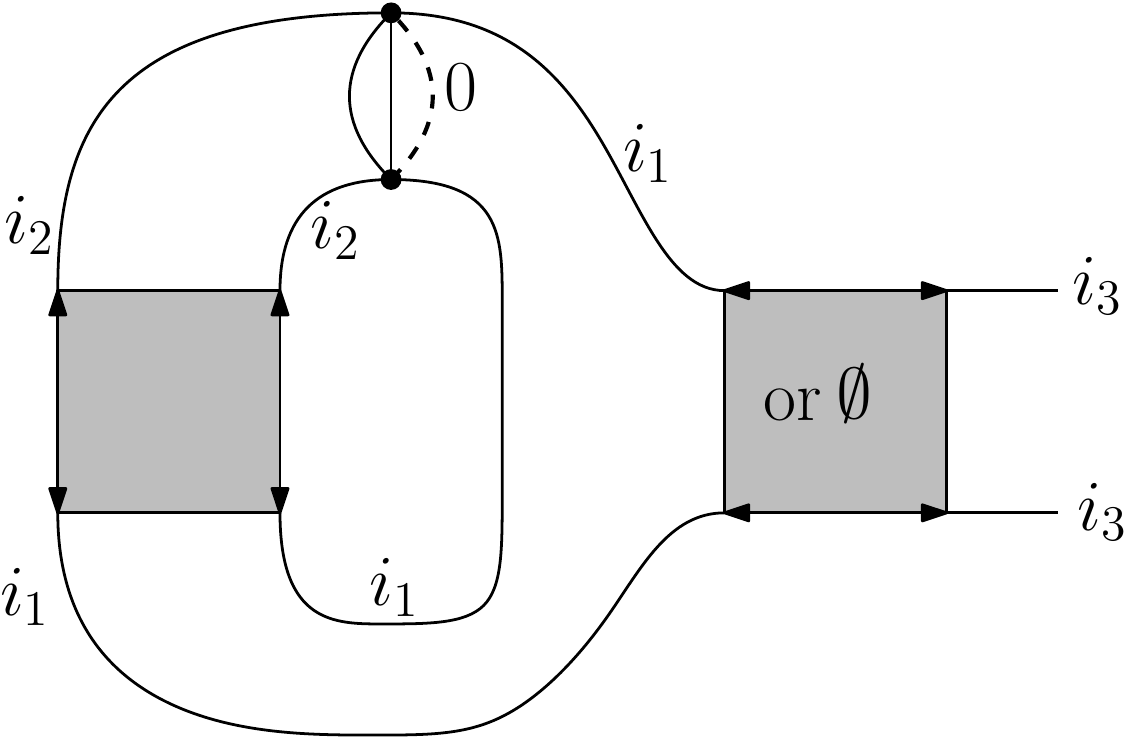} \end{array},% \hspace{1cm} 
\tilde{G}_2^{\text{NLO}(1)} = \begin{array}{c} \includegraphics[scale=.45]{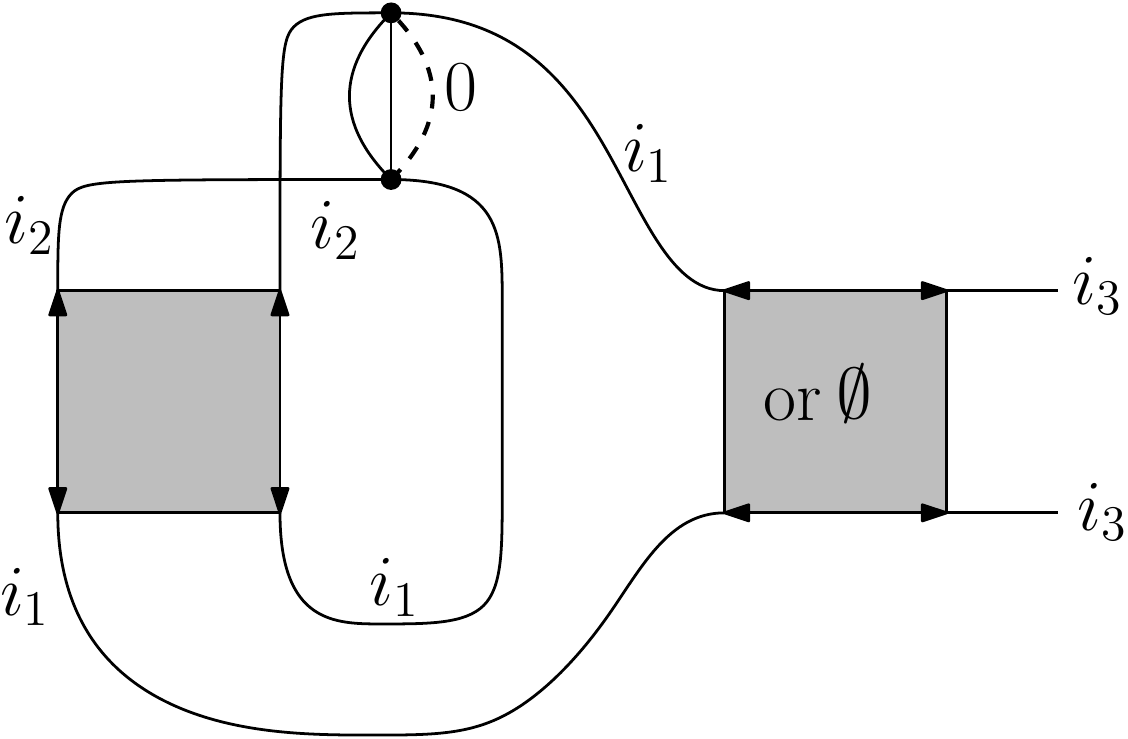} \end{array}
\end{equation}
Only $G_2^{\text{NLO}(1)}$ is bipartite (and would thus contribute in a complex model).

If the cut edge is within a blob of type B, then the NLO 2-point contributions are
\begin{equation} \label{NLO2Pt2}
G_2^{\text{NLO}(2)} = \begin{array}{c} \includegraphics[scale=.45]{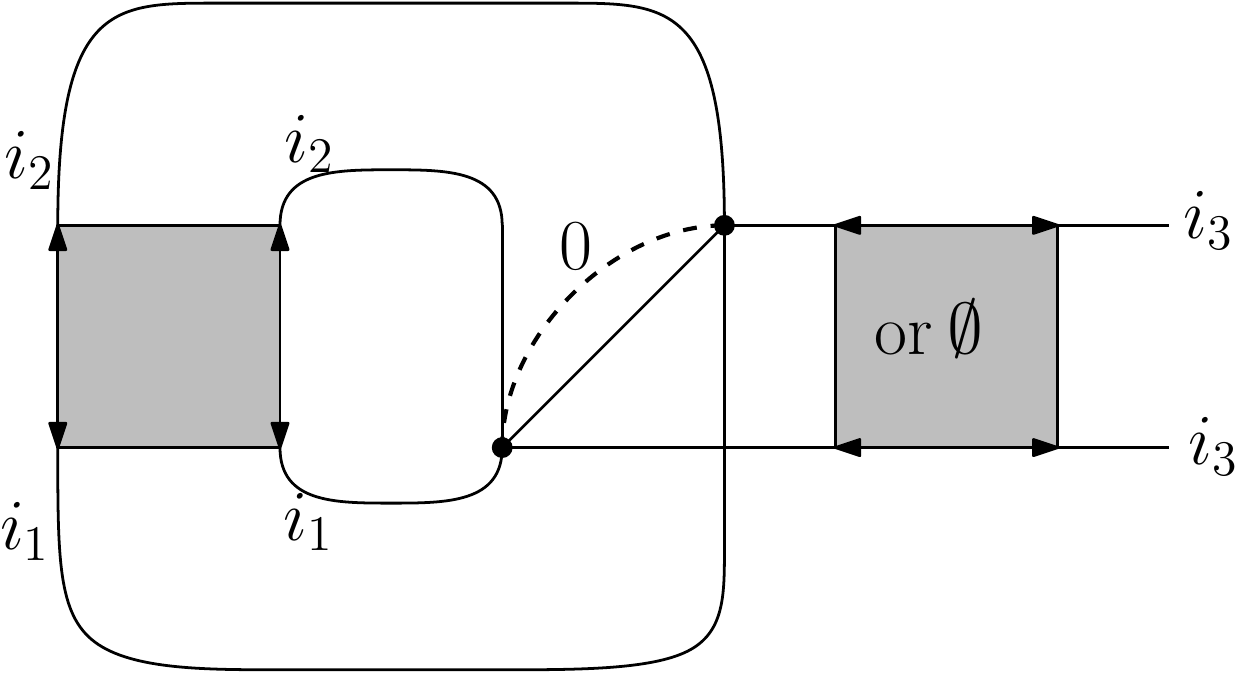} \end{array},% \hspace{1cm}
\tilde{G}_2^{\text{NLO}(2)} = \begin{array}{c} \includegraphics[scale=.45]{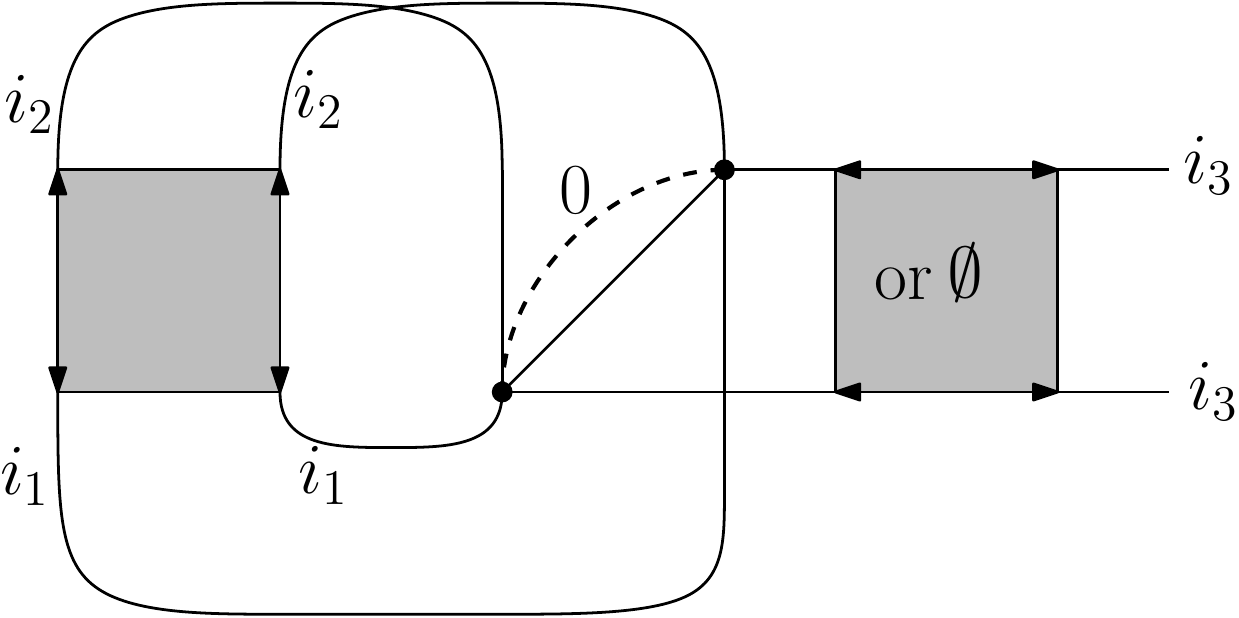} \end{array}
\end{equation}
Again, only $G_2^{\text{NLO}(2)}$ is bipartite.

{ Let us give a specific example of a Feynman diagram which can be obtained as a particular case of $G_2^{\text{NLO}(2)}$ above:
\begin{equation} \label{referee}
\begin{array}{c} \includegraphics[scale=.45]{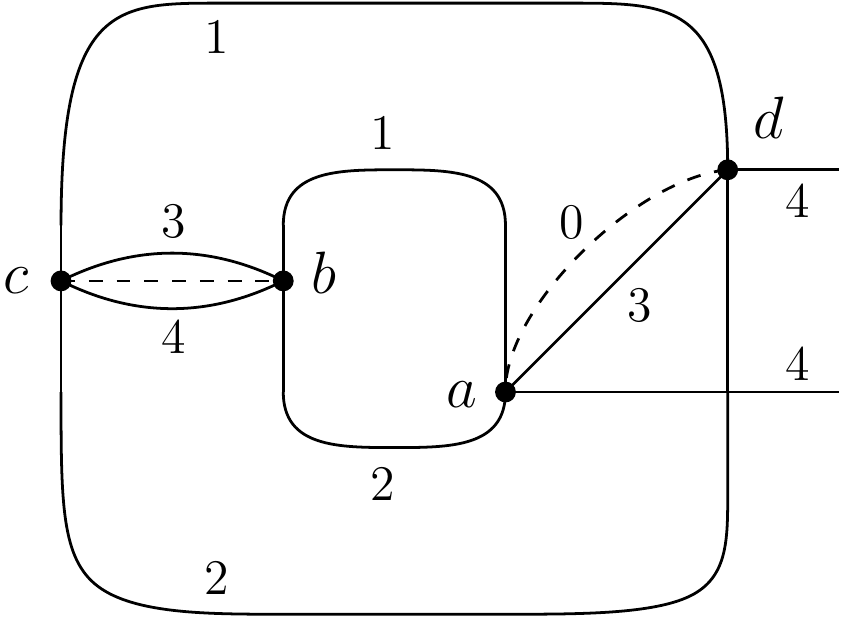} \end{array} = %  \hspace{3cm}
\begin{array}{c} \includegraphics[scale=.75]{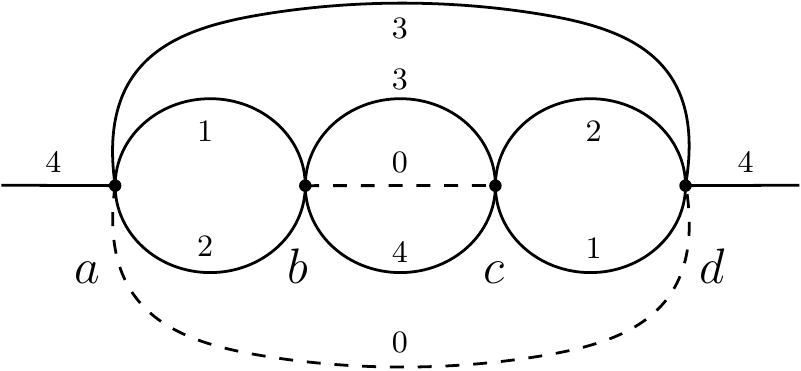} \end{array}
\end{equation}
Thus, on the LHS of \eqref{referee} one has the Feynman diagram obtained if the chain on the LHS of $G_2^{\text{NLO}(2)}$ ({\it i. e.} the chain attached to internal edges) has only two vertices, while the chain on the RHS of $G_2^{\text{NLO}(2)}$  ({\it i. e.}  the chain attached to external edges) is empty. On the RHS of \eqref{referee} we redraw the Feynman diagram thus obtained.
}

\bigskip

%%%%%%%%%%%%%%%%%%%%%%%%%%%%%%%%%%%%%%%%%%%%
%\subsubsection{Following orders of the partition function}
%%%%%%%%%%%%%%%%%%%%%%%%%%%%%%%%%%%%%%%%%%%%

The method we have used to identify LO and NLO contributions to the free energy and 2-point function can in principle be applied at any order. However, the number of diagrams grows importantly and the description becomes tedious. Here we therefore only give the diagrams which contribute to the NNLO of the partition function.

Graphs contributing to the NNLO are such that the corresponding $G_{/0}$ have exactly two independent multicolored cycles,
\begin{equation}
\ell_m(G^{\text{NNLO}}_{/0}) = 2.
\end{equation}
A reasoning similar to the NLO case of Section \ref{sec:2Pt} leads to families of graphs such as the following ones
\begin{equation} \label{VacuumNNLO}
%G^{\text{NNLO}} =
\begin{array}{c} \includegraphics[scale=.4]{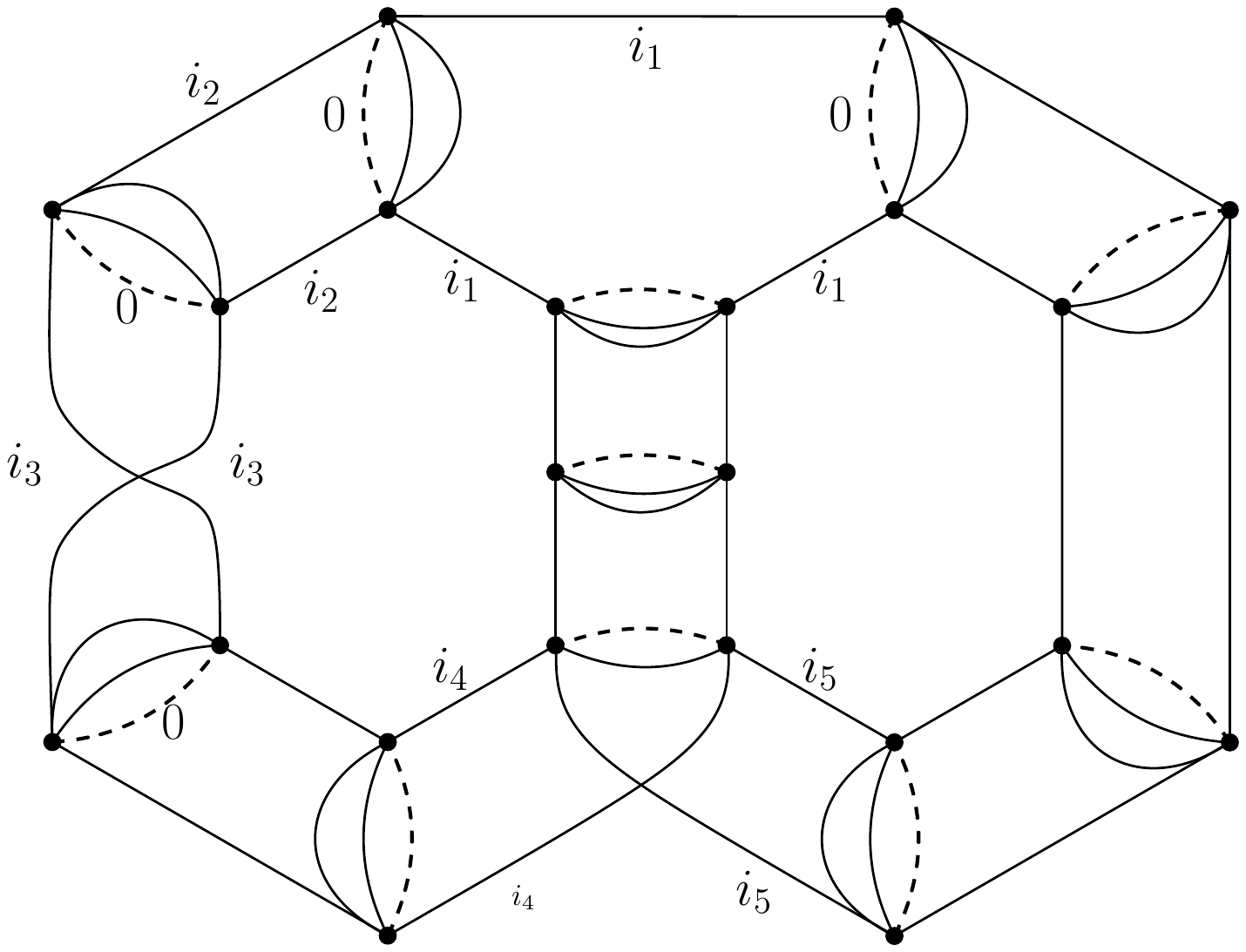} \end{array} \quad \text{and} \quad %\G'^{\text{NNLO}} = 
\begin{array}{c} \includegraphics[scale=.4]{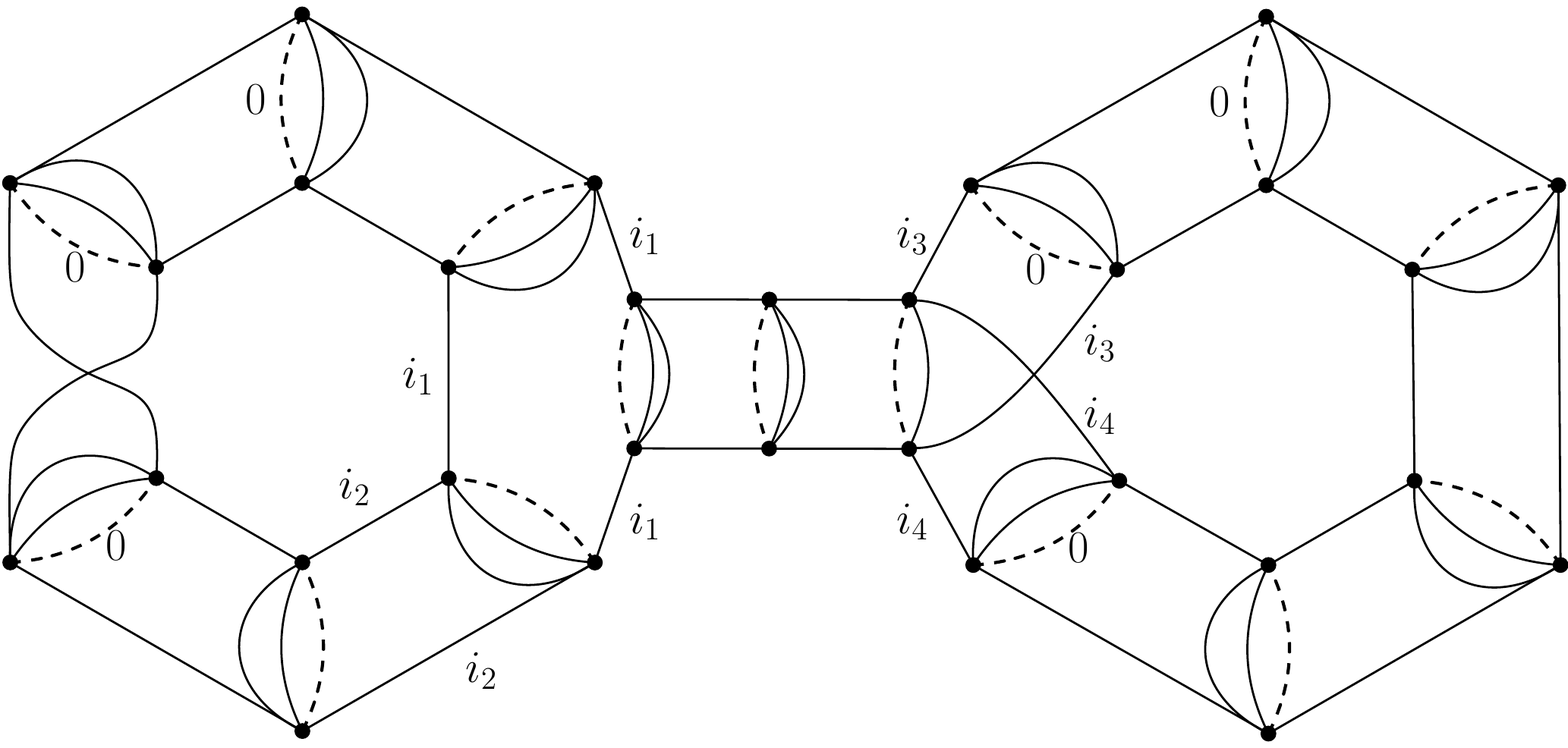} \end{array}
\end{equation}
A collection of diagrams is pictured below. To obtain all such graphs, one has to consider one crossing or no crossing in every loop in every possible way.
\begin{equation}
\label{SYK_Vacuum_NNLO1}
%G_1^{\text{NNLO}} = 
\begin{array}{c} \includegraphics[scale=.35]{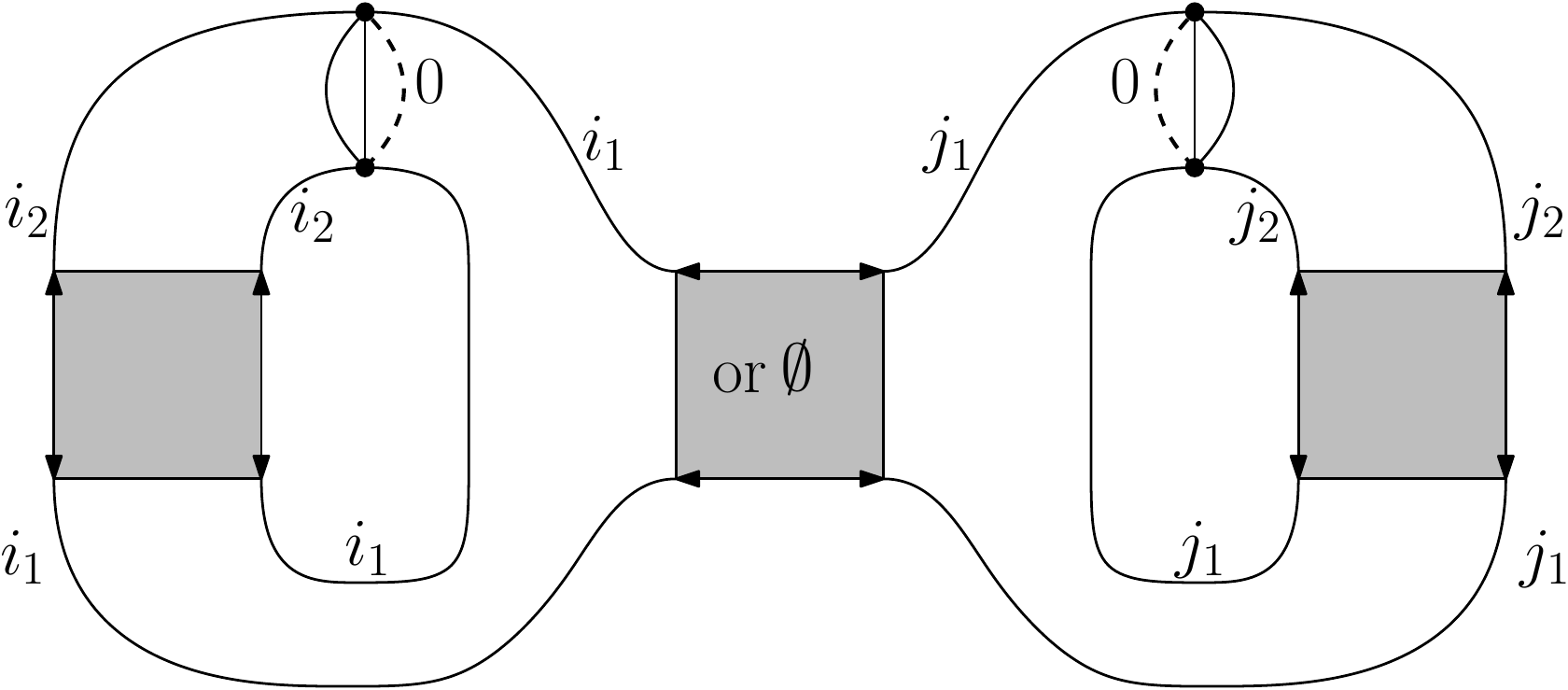} \end{array} \hspace{2.2cm}  %G_2^{\text{NNLO}} = 
\begin{array}{c}
\includegraphics[scale=.35]{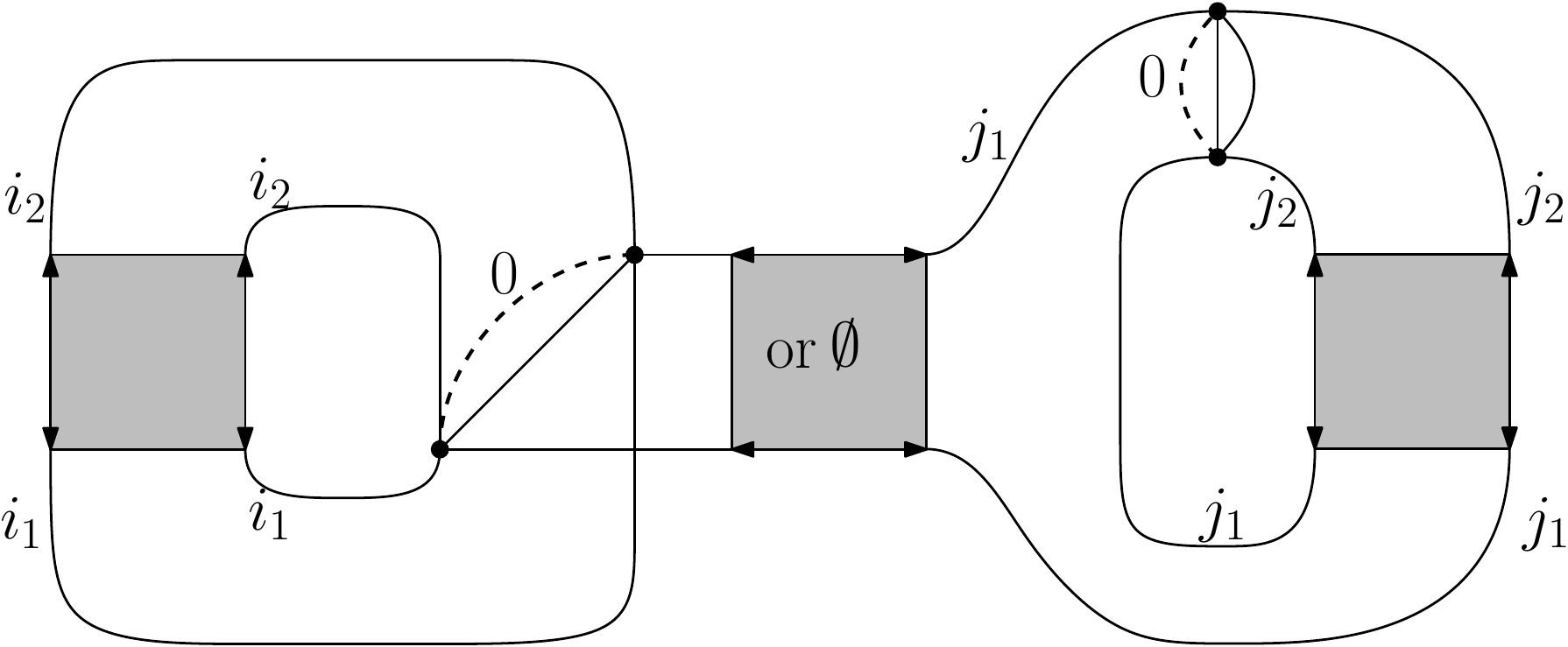} \end{array}
\end{equation}
\begin{equation}
\label{SYK_Vacuum_NNLO2}
%G_3^{\text{NNLO}} = 
\begin{array}{c} \includegraphics[scale=.4]{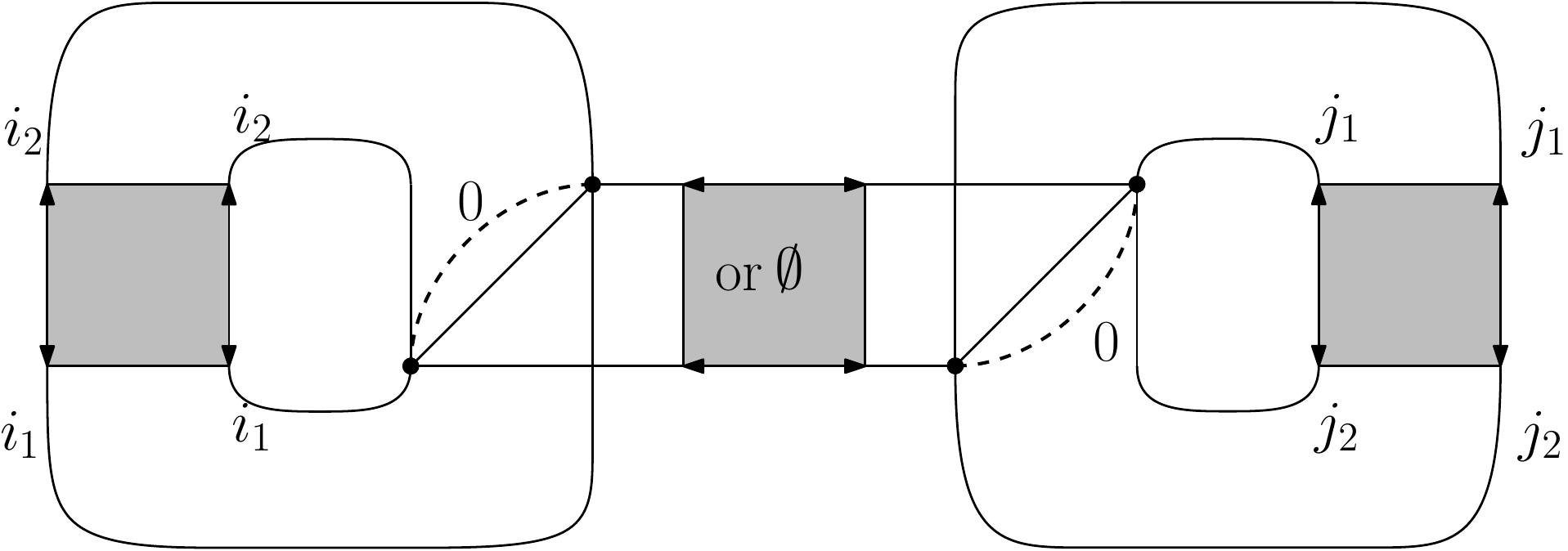} \end{array} \hspace{2cm} %G_4^{\text{NNLO}} = 
\begin{array}{c} \includegraphics[scale=.4]{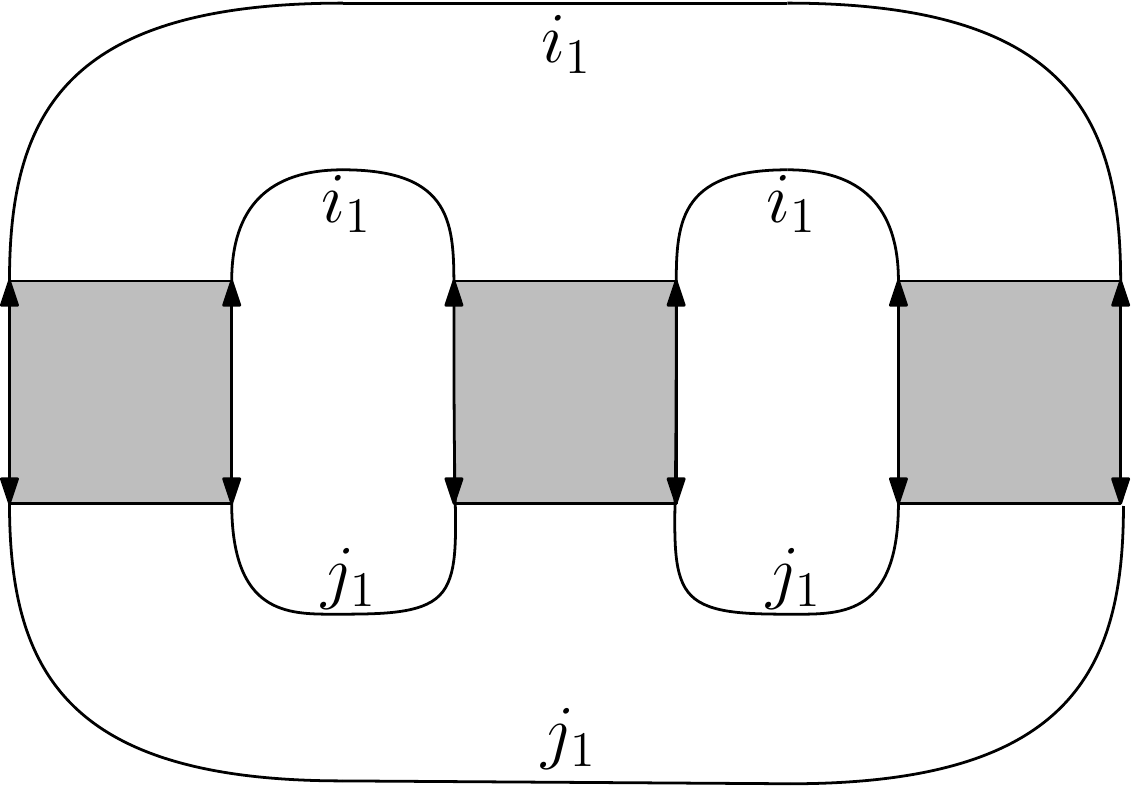} \end{array}
\end{equation}
\begin{equation}
\label{SYK_Vacuum_NNLO3}
%G_5^{\text{NNLO}} = 
\begin{array}{c} \includegraphics[scale=.4]{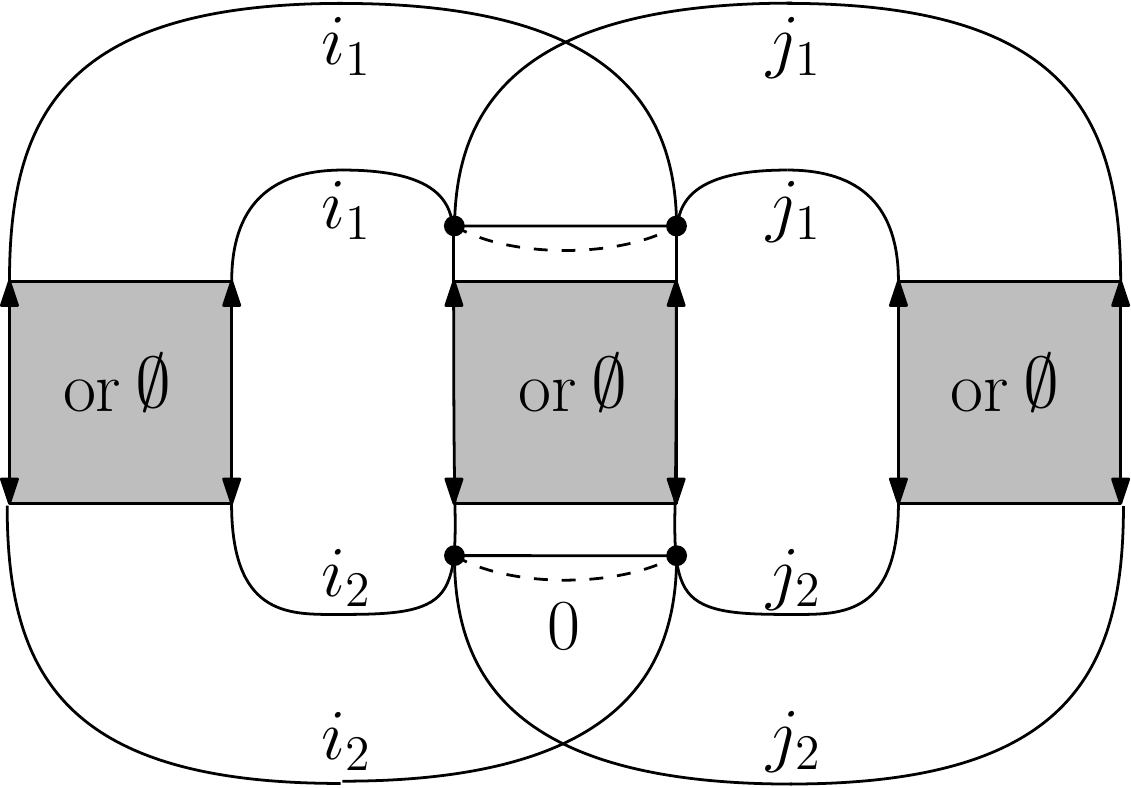} \end{array} \hspace{1cm}  %G_6^{\text{NNLO}} = 
\begin{array}{c} \includegraphics[scale=.4]{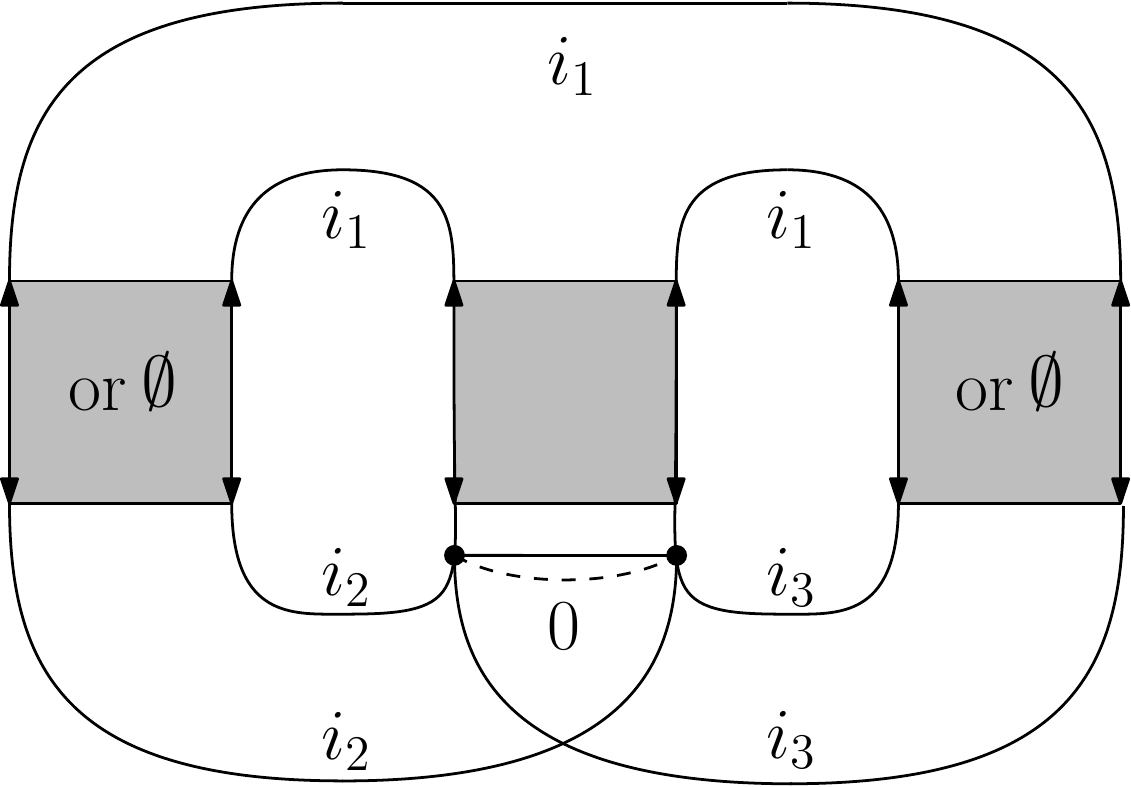} \end{array}
\hspace{1cm} 
\begin{array}{c} \includegraphics[scale=.4]{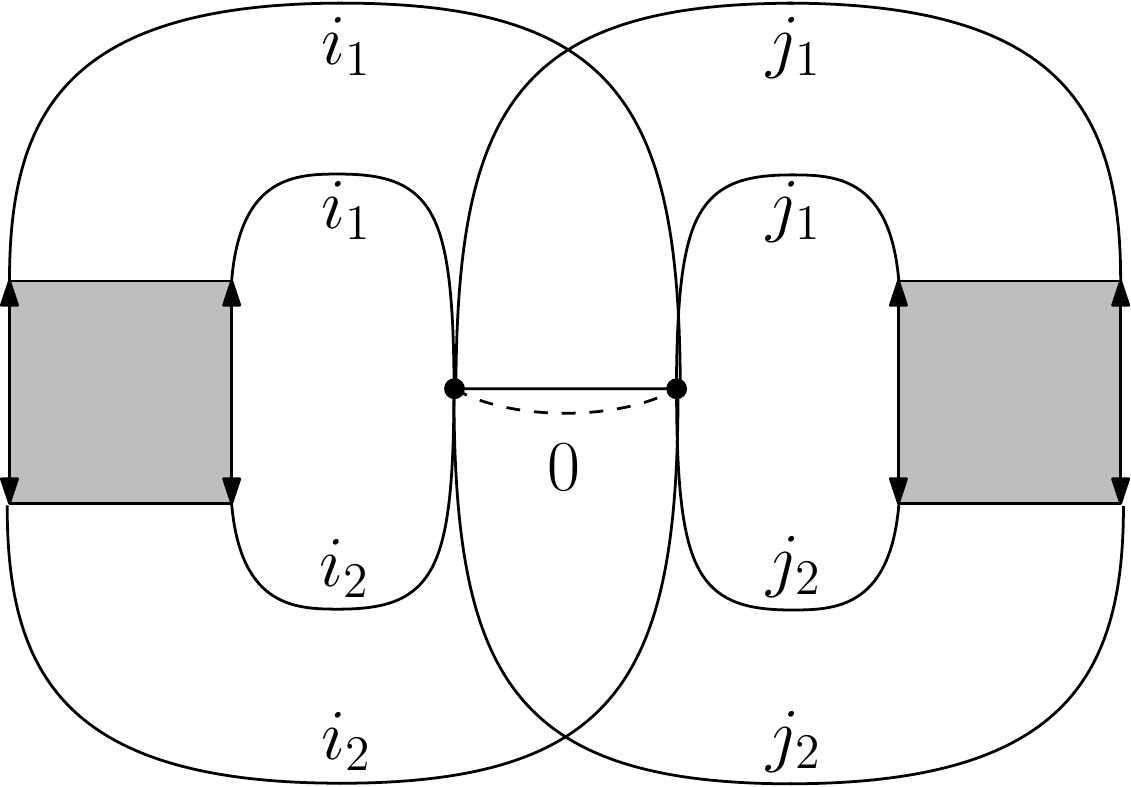} \end{array}
\end{equation}

%%%%%%%%%%%%%%%
\subsubsection{LO and NLO of 4-point functions}
%%%%%%%%%%%%%%%

One can check that the external edges of $4-$point graphs come in pairs where two edges of a pair share the same color. This gives two sets of $2-$point functions, depending on whether all external edges have the same color or not,
\begin{equation}
\langle \psi_i \psi_i \psi_i \psi_i \rangle \quad \text{for $i\in\{1, \dotsc, q\}$, and} \quad \langle \psi_i \psi_i \psi_j \psi_j\rangle \quad \text{for $i\neq j$},
\end{equation}
where $\psi_i, \psi_j$ are fermions of colors $i$ and $j$. Here we have dropped the time dependence since we are only concerned with the diagrammatics. 
%We have also left the vector indices of the fermions implicit, since they are just product of Kronecker deltas which follow easily from the graphs we will give.

There are no major diagrammatic differences between the two types of 4-point functions. We will thus treat both simultaneously.

$4-$point graphs can be obtained by cutting two edges in a vacuum graph. They can be two edges with the same color or two different colors in $\{1, \dotsc, q\}$. If $G$ is a vacuum graph, we denote $G_{e, e'}$ the 4-point graph obtained by cutting $e$ and $e'$. Obviously, if $G_4$ is a 4-point graph, there is a (possibly non-unique) way to glue the external lines two by two, creating two edges $e, e'$,  and to thus get a vacuum graph $G$ such that $G_4 = G_{e,e'}$.

Faces of $G$ and $G_{e, e'}$ are the same except for those which go along $e$ and $e'$. When $e$ and $e'$ have distinct colors, two different faces go along them in $G$ and are thus broken in $G_{e, e'}$. When $e$ and $e'$ have the same color, there can be one or two faces along them. Therefore, the weight received by $G_{e,e'}$ reads
\begin{equation}
w_N(G_{e,e'}) = N^{\chi_0(G_{e,e'})}, \qquad \text{with} \qquad \chi_0(G_{e,e'}) = \chi_0(G) - \eta(G_{e,e'}) \leq 1 - \eta(G_{e,e'}),
\end{equation}
where $\eta(G_{e,e'}) \in\{1, 2\}$ is the number of faces broken by cutting $e$ and $e'$ in $G$.

The classification thus seems a little intricate because of the two possible values for $\eta(G_{e,e'})$. We however claim that it is sufficient to only consider the graphs $G$ with edges $e, e'$ such that
\begin{equation}
\eta(G_{e,e'}) = 2.
\end{equation}

This is always the case when $e$ and $e'$ have different colors. Let us thus focus on the case where $e$ and $e'$ have the same color $i\in\{1, \dotsc, q\}$. Let $G_4$ be a 4-point graph with 4 external legs of color $i$. We claim that there is always one way to connect the external legs pairwise into two edges $e$ and $e'$ with two different faces along them. Denoting $G$ this vacuum graph, we thus interpret $G_4$ as the graph $G_{e,e'}$ with $\eta(G_{e,e'}) = 2$.

With the same notations, we have thus found that
\begin{equation}
\chi_0(G_{e,e'}) = \chi_0(G) - 2.% = -\ell_m(G_{/0}) - 1,
\end{equation}
The strategy is thus for both types of 4-point functions:
\begin{itemize}
\item use the classification of vacuum graphs which we have established: LO, NLO graphs, etc.
\item cut two edges of them such that $\eta(G_{e,e'})=2$.
\end{itemize}

In the large $N$ limit, cutting two edges in melonic graphs (such that $G_{e,e,'}$ remains connected, as well as $G_{e,e'}$ minus its edges of color 0) precisely leads to the chains introduced in \eqref{SYKChain} (one might add 2-point insertions on the external legs).

At NLO, one finds
\begin{equation} \label{NLO4Pt1A}
\begin{aligned}
&A_1 = \begin{array}{c} \includegraphics[scale=.6]{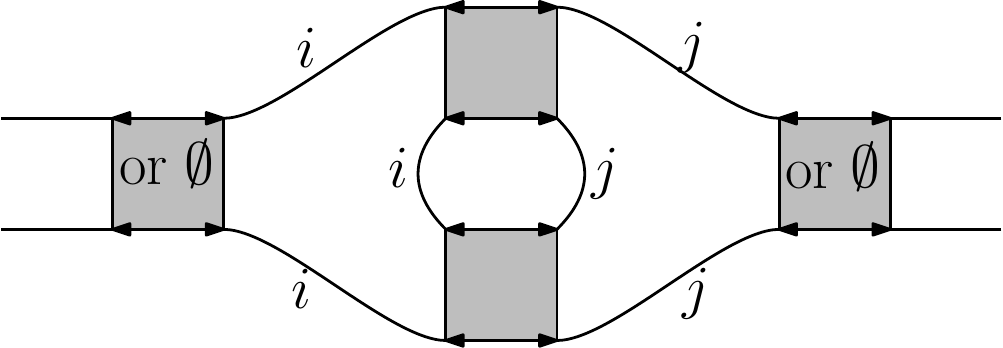} \end{array} & A_2 = \begin{array}{c} \includegraphics[scale=.6]{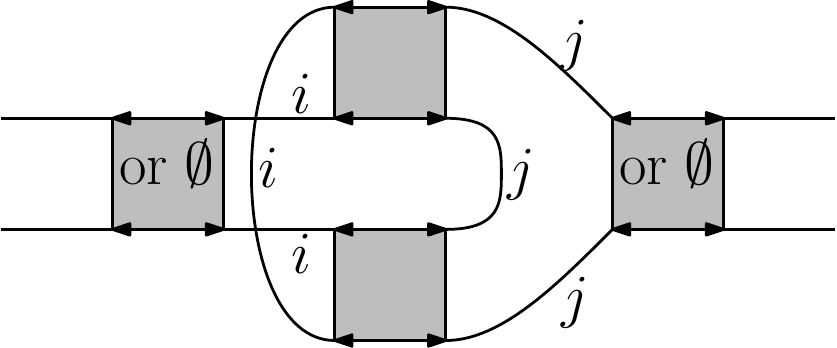} \end{array}\\
&A_3 = \begin{array}{c} \includegraphics[scale=.6]{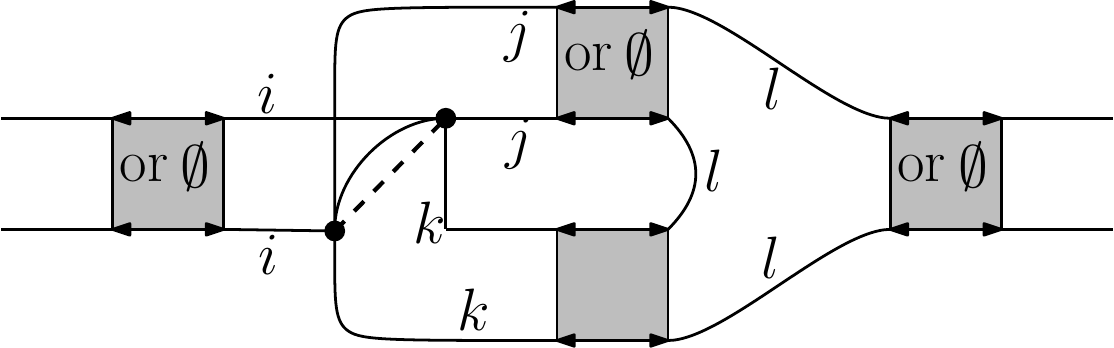} \end{array} & A_4 = \begin{array}{c} \includegraphics[scale=.6]{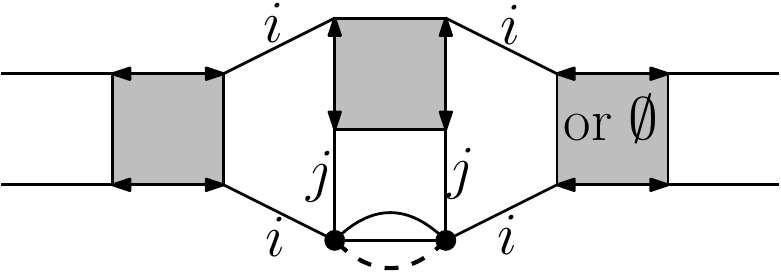} \end{array}
\end{aligned}
\end{equation}
by cutting an edge in $G_2^{\text{NLO}(1)}$ in \eqref{NLO2Pt1},
\begin{equation}
\label{NLO4Pt1B}
\begin{aligned}
&B_1 = \begin{array}{c} \includegraphics[scale=.6]{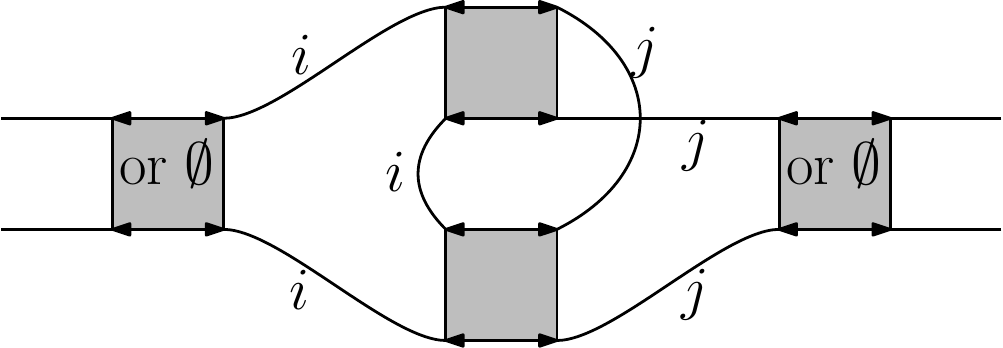} \end{array} & B_2 = \begin{array}{c} \includegraphics[scale=.6]{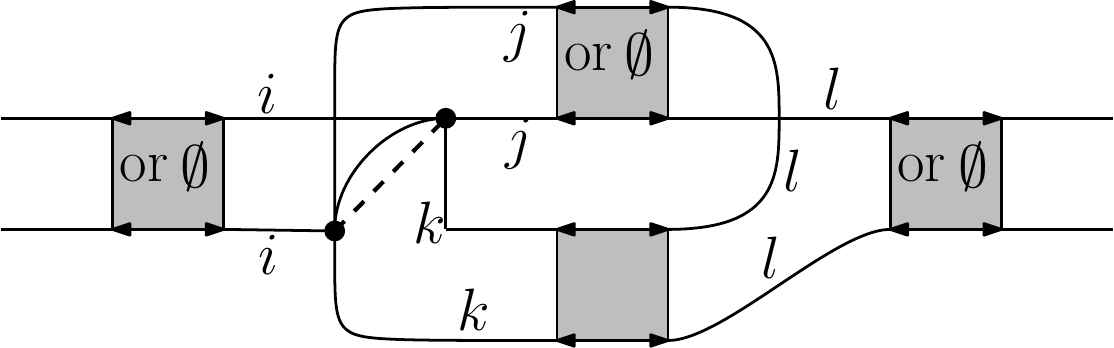} \end{array}\\
&B_3 = \begin{array}{c} \includegraphics[scale=.6]{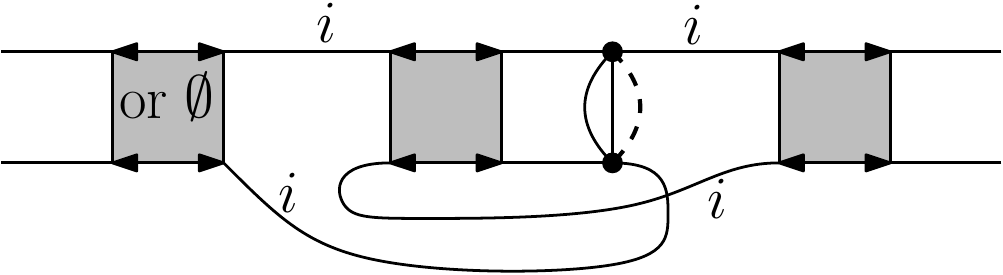} \end{array} & 
\end{aligned}
\end{equation}
by cutting an edge in $\tilde{G}_2^{\text{NLO}(1)}$ in \eqref{NLO2Pt1},
\begin{equation}
\label{NLO4Pt1C}
C_1 = \begin{array}{c} \includegraphics[scale=.6]{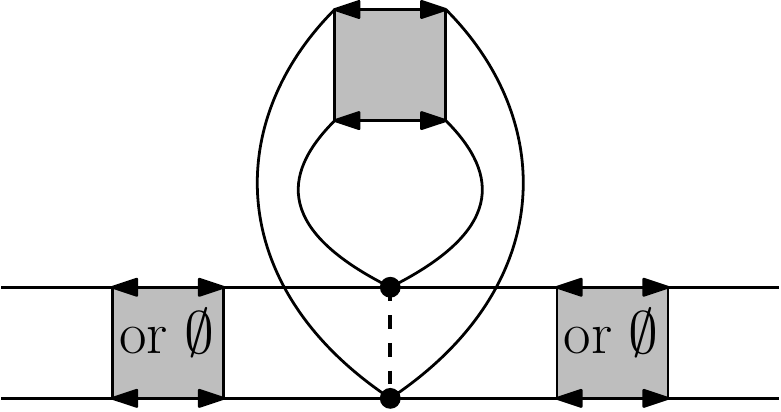} \end{array} \qquad C_2 = \begin{array}{c} \includegraphics[scale=.6]{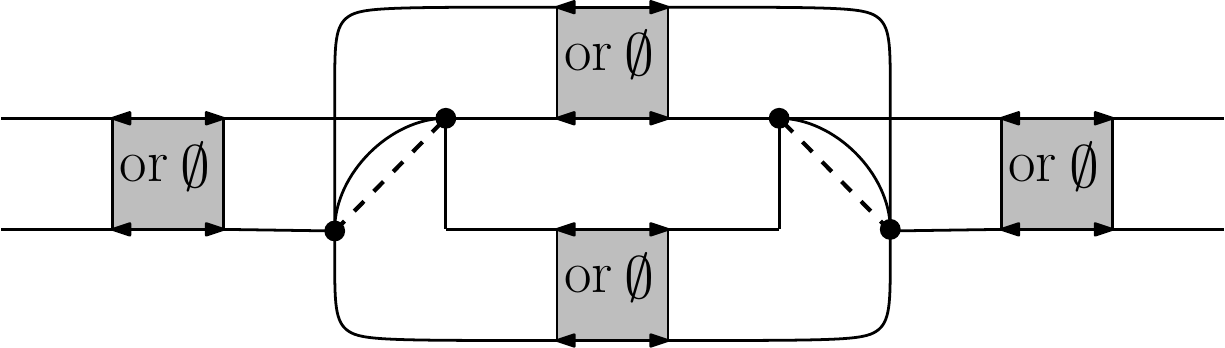} \end{array}
\end{equation}
by cutting an edge in $G_2^{\text{NLO}(2)}$ in \eqref{NLO2Pt2},
\begin{equation}
\label{NLO4Pt1D}
\begin{aligned}
&D_1 = \begin{array}{c} \includegraphics[scale=.6]{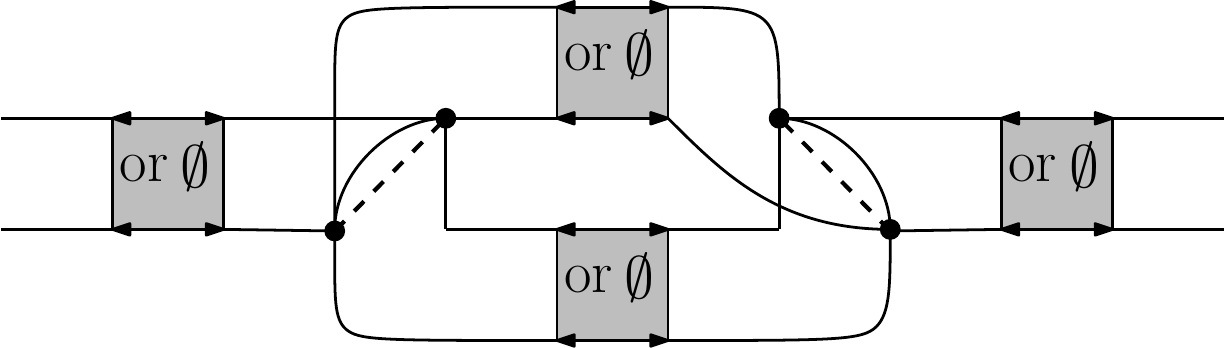} \end{array} & D_2 = \begin{array}{c} \includegraphics[scale=.6]{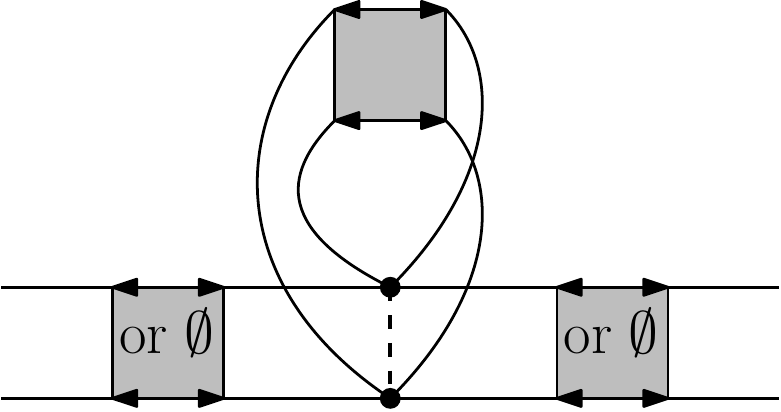} \end{array}
\end{aligned}
\end{equation}
by cutting an edge in $\tilde{G}_2^{\text{NLO}(2)}$ in \eqref{NLO2Pt2}, and finally the two following families
\begin{equation}
\label{NLO4Pt1E}
\begin{array}{c} \includegraphics[scale=.6]{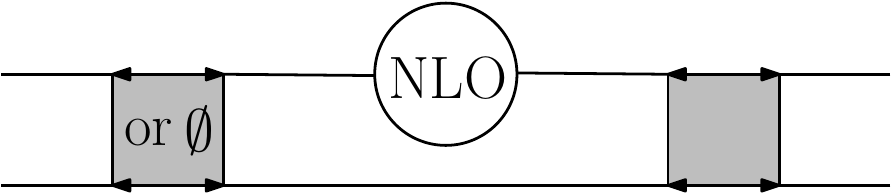} \end{array} \qquad
\begin{array}{c} \includegraphics[scale=.6]{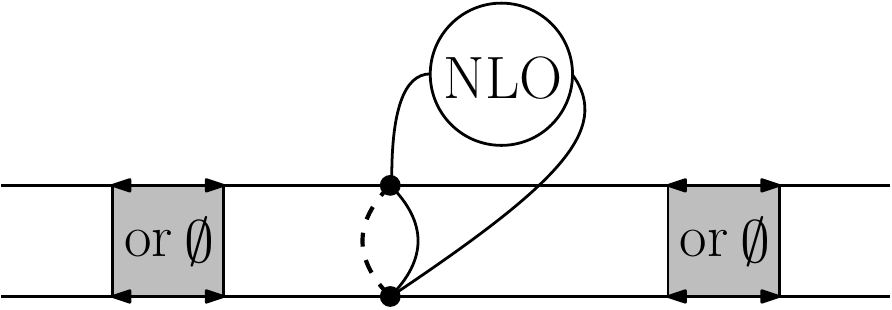} \end{array} 
\end{equation}
obtained by performing a $2-$point insertion of the NLO $2-$point function into the LO $4-$point chains.

The above description avoids redundancies. Notice that the colors are important. For instance, by specializing the middle chains in $A_2$ and $A_4$ to have a single pair of vertices, the same graph is obtained but with different colorings.

\subsection{Non-Gaussian disorder average in the complex model}
\label{sec:nongaus}

A possible generalization of the SYK model is achieved if we consider a non-Gaussian disorder; the quenched disorder for the couplings is given by a non-Gaussian distribution. Consider the complex version of the SYK model containing $q$ flavors with the non-Gaussian disorder, whose action is given by \eqref{act:noir}. Following \cite{TanasaKrajewskiPascalieLaudonio}, one can derive the effective action for this model and show that the effect of this non-Gaussian averaging is a modification of the variance of the Gaussian distribution of couplings at leading order in $N$. Still from \cite{TanasaKrajewskiPascalieLaudonio}, it is possible to prove that the leading order Feynman diagrams are those given by the quadratic term of the distribution (Gaussian universality). This Gaussian universality result for the colored tensor model was initially proved in \cite{universality} and was also exploited in \cite{Bonzom}, in a condensed matter physics setting, to identify an infinite universality class of infinite-range $p-$spin glasses with non-Gaussian correlated quenched distributions.
 
In order to obtain these results, we first need to average the partition function over the non-Gaussian disorder. The most convenient way to perform this is through the use of  replicas. We thus add an extra replica index $r=1,\dots,n$ to the fermions. One has: 

\begin{align}
\langle \log Z(j)\rangle_{j}=\lim_{n\rightarrow 0}\frac{\langle Z^{n}(j)\rangle_{j}-1}{n},
\end{align}
with

\begin{align}
Z^{n}(j)=\int\prod_{1\leq r\leq n}[d\psi_{r}][d\overline{\psi}_{r}]\exp\sum_{r}S_{j}(\psi_{r},\overline{\psi}_{r}) ,
\end{align}
where $S_{j}(\psi,\overline{\psi})$ is given by \eqref{act:noir}. The angle brackets stand for the averaging over $j$, which is performed with a non-Gaussian weight of the type

\begin{align}
\langle Z^{n}(j)\rangle_{j}=\frac{\int dj d\overline{j}\, Z^{n}(j)\exp \big[-\big[\frac{N^{q-1}}{\sigma^{2}}j\overline{j}+V_{N}(j,\overline{j})\big]\big]}{\int djd\overline{j} \exp\big[-\big[\frac{N^{q-1}}{\sigma^{2}}j\overline{j}+V_{N}(j,\overline{j})\big]\big]}.
\end{align}
We further impose that the potential $V_N$ is invariant under independent unitary transformations: 

\begin{align}
j_{i_{1},\dots,i_{q}}\rightarrow \sum_{j_{1},\dots,j_{q}}U^{1}_{i_{1}j_{1}}\cdots U^{q}_{i_{q}j_{q}}j_{j_{1},\dots,j_{q}, }\qquad
\overline{j}_{i_{1},\dots,i_{q}}\rightarrow \sum_{j_{1},\dots,j_{q}}\overline{U}^{1}_{i_{1}j_{1}}\cdots \overline{U}^{q}_{i_{q}j_{q}}\overline{j}_{j_{1},\dots,j_{q}}.\label{unitary}
\end{align}
Assuming that the potential $V_N$ is a polynomial (or an analytic function) in the couplings $j$ and $\overline{j}$, this invariance imposes that the potential can be expanded over non necessarily connected graphs. These graphs are made up by black and white vertices of valence $q$, whose edges connect only black to white vertices (bipartite graphs) and are labeled by a color $a=1,\ldots,q$ in such a way that, at each vertex, the $q$ incident edges carry distinct colors (we thus have edge-colored graphs). The construction of such graphs has already been explained in Sec. \ref{sec:defcolored} and each one can be denoted by a particular contraction of the tensors $j$ and $\overline{j}$. The contraction of their indices means that each white vertex carries a tensor $j$, each black vertex a tensor $\overline{j}$ and that the indices have to be contracted by identifying two indices on both sides of an edge, the place of the index in the tensor being defined by the color of the edge denoted by $c(e)$. We will refer to the graph $G$, using the shorthand $\langle j,\overline{j}\rangle_{ G}$, which is given by

\begin{align}
\langle j,\overline{j}\rangle_{ G}=\sum_{1\leq i_{v,a},\dots,i_{\overline{v},a}\leq N}
\prod_{\text{white}\atop\text{vertices }v}j_{i_{v,1},\dots,i_{v,q}} 
\prod_{\text{black}\atop\text{vertices }\overline{v}}\overline{j}_{\overline{i}_{\overline{v},1},\dots,\overline{i}_{\overline{v},q}} 
\prod_{\text{edges }\atop e=(v,\overline{v})}\delta_{i_{v,c(e)},i_{\overline{v},c(e)}}.\label{graphs}
\end{align}
The most general form of the potential $V_N$ is then expanded over these graphs as:

\begin{align}
V_{N}(j,\overline{j})=\sum_{\text{graph $ G$}}\lambda_{ G}\frac{N^{q-k( G)}}
{\text{Sym($ G$)}}
\langle j,\overline{j}\rangle_{ G} .
\end{align}
In this expression, $\lambda_{ G}$ is a real number, $k( G)$ is the number of connected components of $ G$ and Sym($ G$) its symmetry factor. 
The Gaussian term  corresponds to a dipole graph (a white vertex and a black vertex, connected by $q$ lines) and reads

\begin{align}
\frac{N^{q-1}}{\sigma^{2}}j\overline{j}=
\frac{N^{q-1}}{\sigma^{2}}\sum_{1\leq i_{1},\dots,i_{q}\leq N}
j_{ i_{1},\dots,i_{q}}\overline{j}_{ i_{1},\dots,i_{q}}
\end{align}
Introducing the pair of complex conjugate tensors $K$ and $\overline{K}$ defined by

\begin{align}
K_{ i_{1},\dots,i_{q}}=\text{i}^{\frac q2}\sum_{r}\int dt 
\psi^{1}_{i_{1},r}\cdots \psi^{q}_{i_{q};r}
\qquad
\overline{K}_{ i_{1},\dots,i_{q}}=\text{i}^{\frac q2}\sum_{r}\int dt 
\overline{\psi}^{1}_{i_{1},r}\cdots \overline{\psi}^{q}_{i_{q};r},\label{definitionK}
\end{align}
the averaged partition function reads 

\begin{align}
\label{average}
\langle Z^{n}(j)\rangle_{G}=\frac{\int[d\psi][d\overline{\psi}]
\exp\big[-\int dt \sum_{a,i_{a}}\overline{\psi}^{a}_{i_{a}}\partial_{t}\psi^{a}_{i_{a}}\big]\,\int dj d\overline{j}
\exp\big[-\big[\frac{N^{q-1}}{\sigma^{2}}j\overline{j}+V_{N}(j,\overline{j})+j\overline{K}+\overline{j}K\big]
\big]}{\int djd\overline{j} \exp \big[-\big[\frac{N^{q-1}}{\sigma^{2}}j\overline{j}+V_{N}(j,\overline{j})\big]\big]}.
\end{align}
In order to study the large $N$ limit of the average \eqref{average}, we introduce the background fields $L=-\frac{\sigma^{2}}{N^{q-1}}K$ and $\overline{L}=-\frac{\sigma^{2}}{N^{q-1}}\overline{K}$. Let us shift the variables $j$ and $\overline{j}$ by the background fields $L$ and $\overline{L}$. The numerator in the integral \eqref{average} reads

\begin{align}
\exp\bigg[-\frac{\sigma^{2}}{N^{q-1}}K\overline{K}\bigg]
\int dj d\overline{j}
\exp-\bigg[\frac{N^{q-1}}{\sigma^{2}}j\overline{j}+V_{N}\Big(j-\frac{\sigma^{2}}{N^{q-1}}K,\overline{j}-\frac{\sigma^{2}}{N^{q-1}}\overline{K}\Big)\bigg] 
\end{align}
and the effective potential in the shifted variables is

\begin{align}
V_{N}(s,L,\overline{L})=-\log
\int dj d\overline{j}
\exp-\bigg[\frac{N^{q-1}}{s}j\overline{j}+V_{N}\Big(j+L,\overline{j}+\overline{L}\Big)\bigg]\quad+N^{q}\log\frac{\pi s}{N^{q-1}}\label{effectivepotential}
\end{align}
In this framework, $s$ is a parameter that interpolates between the integral we have to compute, at $s= \sigma^{2}$ (up to a trivial multiplicative constant) and the potential we started with at $s=0$ (no integration and $j=\overline{j}=0$). 
The inclusion of the constant ensures that the effective potential remains zero when we start with a vanishing potential. %In other words
This comes to:

\begin{align}
\int dj d\overline{j}
\exp\bigg[-\bigg[\frac{N^{q-1}}{\sigma^{2}}j\overline{j}+V_{N}\Big(j-\frac{\sigma^{2}}{N^{q-1}}K,\overline{j}-\frac{\sigma^{2}}{N^{q-1}}\overline{K}\Big)\bigg]\bigg]\nonumber\\
=
\bigg(\frac{N^{q-1}}{\pi s}\bigg)^{N^{q}}
\exp \bigg[-V_{N}\bigg(s=\sigma^{2},L=-\frac{\sigma^{2}}{N^{q-1}}K,\overline{L}=-\frac{\sigma^{2}}{N^{q-1}}\overline{K}\bigg)\bigg].
\end{align}
After having performed the average over the non-Gaussian disorder and derived the effective potential, we use a Polchinski-like flow equation to show the Gaussian universality following the approach proposed in \cite{Krajewski} for tensor models and group field theory (see also \cite{Krajewski2}, \cite{Krajewski3} and \cite{Krajewski4}).

Using standard QFT manipulations (see for example, the book \cite{book-ZJ}), one can show that the effective potential $V_{N}(s,L,\overline{L})$ in eq. \eqref{effectivepotential} obeys the following differential equation:

\begin{align}
\frac{\partial V}{\partial s}=\frac{1}{N^{q-1}}\sum_{1\leq i_{1},\dots,i_{q}\leq N}\bigg(
\frac{\partial ^{2}V}{\partial L_{i_{1},\dots,i_{q}}\partial \overline{L}_{i_{1},\dots,i_{q}}}-
\frac{\partial V}{\partial L_{i_{1},\dots,i_{q}}}\frac{\partial V}{\partial \overline{L}_{i_{1},\dots,i_{q}}}
\bigg)\label{ERGE}
\end{align}
One can represent this equation in a graphical way as shown in Fig. \ref{fig:polchinski}. The first term on the RHS corresponds to an edge closing a loop in the graph and the second term in the RHS corresponds to a bridge (also known as 1PR) edge or $q$-cut (see Def. \ref{def:2Cut} of a $2$-cut, the generalization to the $q$-cut is trivial).
\begin{figure}
    \centering
    \includegraphics{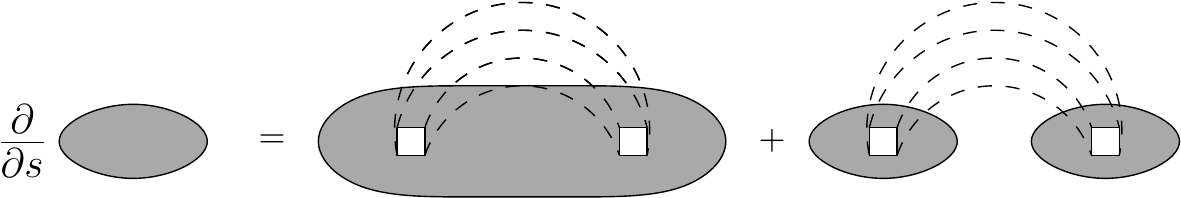}
    \caption{Graphical representation of equation \eqref{ERGE} for $q=4$.}
    \label{fig:polchinski}
\end{figure}
This equation is formally a Polchinski-like equation \cite{Polchinski84}, albeit there are no short distance degrees of freedom over which we integrate. In our context it simply describes a partial integration with a weight $s$ and will be used to control the large $N$ limit of the effective potential.

Since the effective potential is also invariant under the unitary transformations defined in eq. \eqref{unitary}, it may also be expanded over graphs as in \eqref{graphs},
\begin{align}
V_{N}(s,L,\overline{L})=
\sum_{\text{graph $ G$}}\lambda_{ G}(s)\frac{N^{q-k(q)}}
{\text{Sym($ G$)}}\langle L,\overline{L}\rangle_{ G},
\label{effectivegraphs}
\end{align}
with $s$ dependent couplings $\lambda_{ G}(s)$. Inserting this graphical expansion in the differential equation \eqref{ERGE}, we obtain a system of differential equations for the couplings,
\begin{align}
\frac{d\lambda_{ G}}{ds}=\sum_{ G'/(\overline{v}v)= G}
N^{k( G)-k( G')+e(v,\overline{v})-q+1}\,\lambda_{ G'}-
\sum_{( G'\cup G'') /(\overline{v}v)= G}\lambda_{ G'}\,\lambda_{ G''}
\label{system}
\end{align}
A derivation of the potential $V_N$ with respect to $L_{i_{1},\dots,i_{q}}$ (resp. $\overline{L}_{i_{1},\dots,i_{q}}$) removes a white vertex (resp. a black vertex). Then, the summation over the indices  in $i_{1},\dots,i_{q}$ in \eqref{ERGE} reconnects the edges, respecting the colors.  

In the first term on the RHS of \eqref{ERGE}, given a graph $ G$ in the expansion of the LHS, we have to sum over all graphs $ G'$ and pairs of a white vertex $v$ and a black vertex $\overline{v}$ in $ G'$ such that the graph $ G'/(\overline{v}v)$ obtained after reconnecting the edges (discarding the connected components made of single lines) is equal to $ G$ - see Fig. \ref{remove} and Fig. \ref{remove2}.

\begin{figure}[ht]
\centering
\includegraphics[scale=0.5]{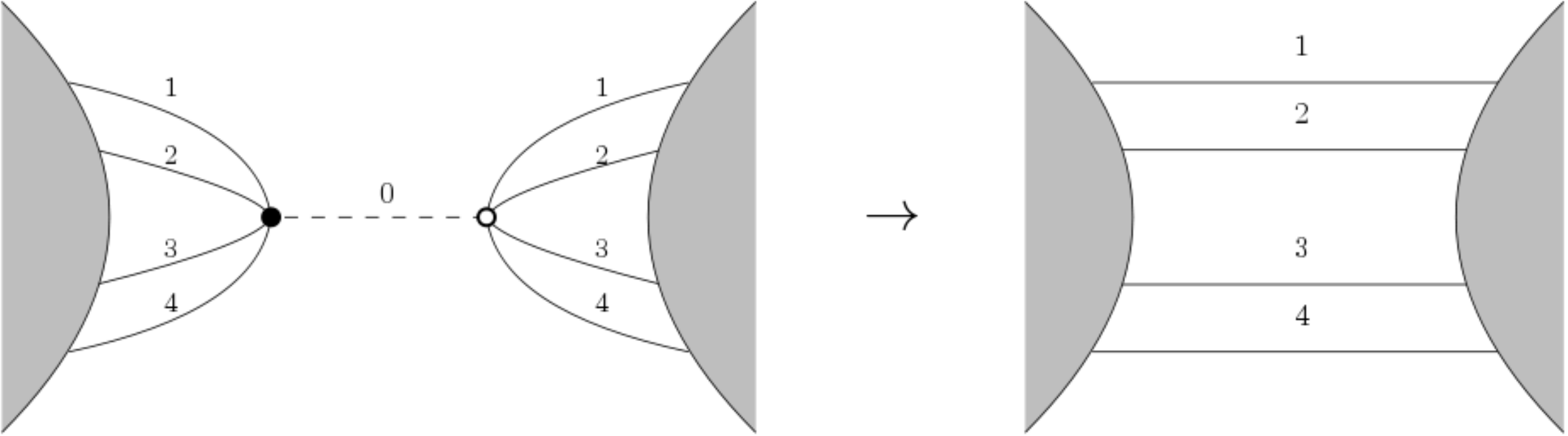}
\caption{Removal of a white and a black vertex and re-connection of the edges.}
\label{remove}
\end{figure}

\begin{figure}[ht]
\centering
\includegraphics[scale=0.5]{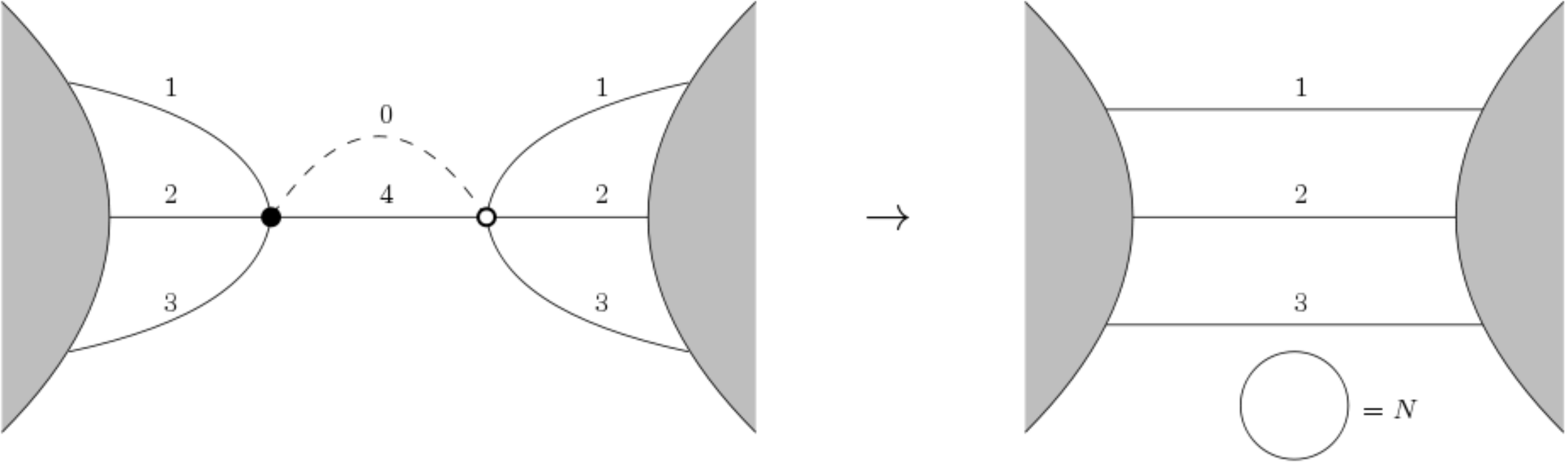}
\caption{Removal of a white and a black vertex and re-connection of the edges creating a loop.}
\label{remove2}
\end{figure}

The number $e(v,\overline{v})$ is the number of edges directly connecting $v$ and $\overline{v}$ in $ G$. After summation over the indices, each of these lines yields a power of $N$, which gives the factor  of $N^{e(v,\overline{v})}$. 

The operation of removing two vertices and reconnecting the edges can at most increase the number of connected components (including the graphs made of single closed lines) by $q-1$, so that we always have $k( G)-k( G')+e(v,\overline{v})-q+1\leq 0$. We obtain the equality if and only if $ G'$ is a melonic graph. 
Therefore, in the large $N$ limit, only melonic graphs survive in the first term on the RHS of \eqref{system} (this is further proof of the melonic dominance in the SYK model already shown in Theorem \ref{thm}).

In the second term, we sum over graphs $ G'$ and  white vertices $v\in G'$ and graphs $ G''$ and black vertices $\overline{v}\in G''$, with the condition that the graph obtained after removing the vertices and  reconnecting the lines $( G'\cup G'')/(\overline{v}v)$ is equal to $ G$. In that case, the number of connected components necessarily diminishes by $1$, so that all powers of $N$ cancel. \\
The crucial point in the system \eqref{system} is that only negative (or null) powers of $N$ appear. It can be written as
\begin{align}
\frac{d\lambda_{ G}}{ds}=
\beta_{0}\big(\left\{\lambda_{ G}\right\}\big)
+\frac{1}{N}\beta_{1}\big(\left\{\lambda_{ G}\right\}\big)
+\dots
\end{align}
As a consequence, if $\lambda_{ G}(s=0)$ is bounded, then $\lambda_{ G}(s)$ is also bounded for all $s$ (i.e. it does not contain positive powers of $N$).

Let us now substitute $L=-\frac{\sigma^{2}}{N^{q-1}}K$ and $\overline{L}=-\frac{\sigma^{2}}{N^{q-1}}\overline{K}$ in the expansion of the effective potential
\eqref{graphs},
\begin{align}
\label{final}
V_{N}\bigg(s=\sigma^{2},L=-\frac{\sigma^{2}}{N^{q-1}}K,\overline{L}=-\frac{\sigma^{2}}{N^{q-1}}\overline{K}\bigg)=
\sum_{\text{graph $ G$}}\lambda_{ G}(\sigma^{2})\frac{(-\sigma^{2})^{v( G)}N^{q-k(q)-(q-1)v( G)}}
{\text{Sym($ G$)}}\langle K,\overline{K}\rangle_{ G}.
\end{align}
Here $v( G)$ is the number of vertices of $ G$. 
The exponent of $N$ can be rewritten as $(q-1)(1-v( G))+1-k( G)$. It has it maximal value for $v( G)=2$ and $k( G)=1$, which corresponds to the dipole graph. 
This is a re-expression the Gaussian universality property of random tensors.

Taking into account the non-Gaussian quenched disorder, we now derive the effective action for the bilocal invariants

\begin{align}
\widetilde{G}_{r,r'}^{a}(t,t')=\frac{1}{N}\sum_{i}\psi^{a}_{i,r}(t_{1})\overline{\psi}^{a}_{i,r'}(t').\label{bilocal}
\end{align}
Note that these invariants carry one flavour label $a$ and two replica indices $r,r'$. \\
To this end, let us come back to the partition function \eqref{average}. We then express the result of the average over $j$ and $\overline{j}$ as a sum over graphs $ G$ using the expansion of the effective potential \eqref{final} and replacing the tensors $K$ and $\overline{K}$ in terms of the fermions $\psi$ and $\overline{\psi}$ (see eq. \eqref{definitionK}). \\
Then, each graph $ G$ involves the combination

\begin{align}
\langle K,\overline{K}\rangle_{ G}=\sum_{1\leq i_{v,a},\dots,i_{\overline{v},a}\leq N}
&\prod_{\text{white}\atop\text{vertices }v}\sum_{r_{v}}\int dt_{v}
\psi^{1}_{i_{v,1},r_{v}}(t_{v})\cdots\psi^{q}_{i_{v,q},r_{v}}(t_{v})\nonumber \\
&\prod_{\text{black}\atop\text{vertices }\overline{v}}
\sum_{r_{\overline{v}}}\int dt_{\overline{v}}
\overline{\psi}^{1}_{\overline{i}_{\overline{v},1},r_{\overline{v}}}\cdots
\overline{\psi}^{q}_{\overline{i}_{\overline{v},q},\overline{v}} (t_{\overline{v}})\prod_{\text{edges }\atop e=(v,\overline{v})}\delta_{i_{v,c(e)},i_{\overline{v},c(e)}}.
\label{graphs2}
\end{align}
After introducing the Lagrange multiplier $\widetilde{\Sigma}$ to enforce the constraint \eqref{bilocal} and assuming a replica symmetric saddle-point, the effective action of our model writes:
\begin{align}
   \frac{\mathcal{S}_{eff}[\mathrm{G},\Sigma]}{N} = &- \sum \limits_{f=1}^{q} \log{\det{\big(\delta(t_1-t_2)\partial_{t}-\widetilde{\Sigma}_f(t_1,t_2) \big)}}
    +\int \mathrm{d}\mathbf{t} \sum \limits_{f=1}^{4}\widetilde{\Sigma}_f(\mathbf{t}) \widetilde{G}_f(\mathbf{t}) \\
    &-\sum_{ G} N^{-(v( G)-2)(q/2-1)+1-k( G)}\mu_{ G}(\sigma^2,\{\lambda_{ G'}\}) \langle \widetilde{G} \rangle_{ G},
\end{align}

The term $\langle \widetilde{G} \rangle_{ G}$ associated to a graph $ G$ is constructed as follows:
\begin{figure}
    \centering
    \includegraphics{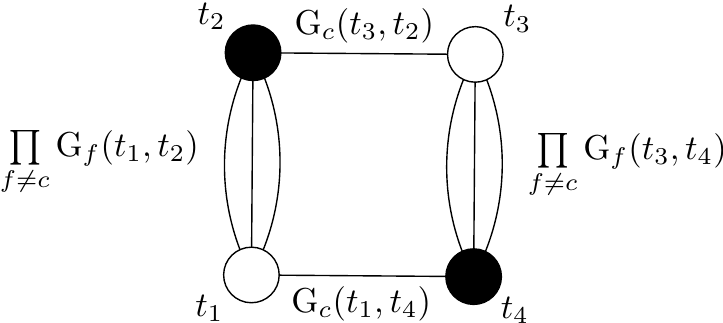}
    \caption{Graphical representation of the term $\langle\mathrm{G}\rangle_{ G}$ for the quartic melonic graph for $q=4$.}
    \label{fig:effpillow}
\end{figure}

\begin{itemize}
\item to each vertex associate a real variable $t_{v}$;
\item to an edge of colour $c$ joining $v$ to $v'$ associate $\widetilde{G}_{c}(t_{v},t_{v'})$;
\item multiply all edge contributions and integrate over vertex variables. 
\end{itemize}

We then add up these contributions, with a weight $\lambda_{ G}$ and a power of $N$ given by
\begin{align}
N^{q-k( G)}\times (N^{-(q-1)})^{v( G)} \times N^{e( G)}=N \times N^{-(v( G)-2)(q/2-1)+1-k( G)} ,
\end{align}
with $e( G)$ the number of edges of $ G$, obeying $2e( G)=qv( G)$. \\
At leading order in $N$, only the Gaussian terms survives ({\it i.e.} the graph $ G$ with $v( G)=2$ and $ k( G)=1$), except for the matrix model case ($q=2$). In this case, all terms corresponding to connected graphs survive. 
Let us emphasize that the variance of the Gaussian distribution of coupling  is thus modified, as a consequence of the non-Gaussian averaging of our model. Remarkably, for $q>2$, this is the only modification at leading order in $N$.

The actual value of the covariance (which we denote by $\sigma'$) induced by non Gaussian disorder is most easily computed using a Schwinger-Dyson equation, see \cite{Bonzom:2012cu}. In our context, the latter arises from

\begin{align}
\sum_{i_{1}\dots i_{q}}
\int dj d\overline{j}\, \frac{\partial}{\partial \overline{j}_{i_{1}\dots i_{q}}} 
\bigg\{
{j}_{i_{1}\dots i_{q}}
\exp-\Big[\frac{N^{q-1}}{\sigma^{2}}j\overline{j}+V_{N}(j,\overline{j})\Big]
\bigg\}=0.
\end{align}
At large $N$, it leads to the algebraic equation
\begin{align}
1=\frac{\sigma'^{2}}{\sigma^{2}}+
    \sum_{\text{melonic graph $ G$}}\frac{\lambda_{ G} }
{\text{Sym($ G$)}}\,(\sigma')^{v( G)}
\end{align}

\medskip

Finally, it is interesting to note that this effective action, despite being non local, is invariant under reparametrization (in the IR) at all orders in $1/N$:
\begin{align}
G(t,t')\rightarrow \bigg(\frac{d\phi}{dt}(t)\bigg)^{\Delta}\bigg(\frac{d\phi}{dt'}(t')\bigg)^{\Delta}G(\phi(t),\phi(t')).
\end{align}
Indeed, changing the vertex variables as $t_{v}\rightarrow \phi(t_{v})$, the jacobians exactly cancel  with the rescaling of $G$ since $\Delta=1/q$ and all vertices are are $q$-valent.

\section{Concluding remarks and perspectives}
\label{sec:conclusion}

In this review paper we first showed the melonic dominance at large $N$ in the SYK model. We then considered  a colored version of the SYK model, which is a particular case of the Gross and Rosenhaus generalization of the SYK model \cite{Gross}, and is the real version of the model studied in \cite{gurau-ultim}. We have then 
analyzed the diagrammatics of the two- and four-point functions of this model, exhibiting the LO and NLO diagrams in the large $N$ expansion. 
%This method can be carried on to higher orders (the main difficulty being then the larger number of diagrams to account for). 
In particular, we have applied it to extract the NLO of the 4-point function, thus going beyond chain diagrams (also called ladders in \cite{maldacena}).
Finally, we have investigated the effects of non-Gaussian average over the random couplings $J$ in a complex clored SYK model. 
%To our knowledge, this is the first study of the effects of the relaxation of the Gaussianity condition in SYK models when no double scaling limit is taken.

A perspective for future work appears to us to be the computation of the corresponding Feynman amplitudes of the NLO diagrams we have shown in this review. Moreover several other SYK-like tensor models exist in the literature (the Klebanov-Tarnopolsky model \cite{CTKT}, or the supersymetric model \cite{brown}). It would thus be interesting to apply the diagrammatic techniques we have developed here for the study of these models as well, in order to compare the LO and NLO behavior of all these SYK-like models.

Another interesting perspective appears to us to be the investigation of the effects of a perturbation from Gaussianity in the case of $q=2$ (fermions with  a random mass matrix) and in the case 
of the real SYK model. The main technical complication in this latter case comes from the fact that one has to deal with graphs which are not necessary bipartite - the removal and reconnection of edges of these graphs (which is the main technical ingredient of our approach) being much more involved. It would thus be interesting to check weather or not in this case also, non-Gaussian perturbation leads to a modification of the variance of the Gaussian distributions of the couplings $J$ at leading order in $N$, as we proved to be the case for the complex version of the SYK model studied here.

\section*{Acknowledgements}
%If you'd like to thank anyone, place your comments here
R. Pascalie is partially supported by the Deutsche Forschungsgemeinschaft (DFG, German Research
Foundation) under Germany's Excellence Strategy EXC 2044--390685587, Mathematics Münster: Dynamics--Geometry--Structure.
A.~Tanasa is partially supported by the 
%CNRS Infiniti ModTens grant
PN 09 37 01 02 grant. 

%\end{acknowledgements}

%\appendix

% BibTeX users please use one of
\bibliographystyle{plain}  
\bibliography{biblioTanasa}  % name your BibTeX data base

\bigskip

\flushright

{\small
M. Laudonio\\
LaBRI, Universit\'e de Bordeaux, CNRS UMR 5800, Talence, France,\\
Department of Applied Mathematics, University of Waterloo, Waterloo, Ontario, Canada\\
\medskip
R. Pascalie\\ 
LaBRI, Universit\'e de Bordeaux, CNRS UMR 5800, Talence, France,\\
Mathematisches Institut der Westfalischen Wilhelms-Universitaat,  M\"unster, Germany, EU\\
\medskip
A. Tanasa \\
LaBRI, Universit\'e de Bordeaux, CNRS UMR 5800, Talence, France,\\
H. Hulubei Nat.  Inst.  Phys.  Nucl.  Engineering, Magurele, Romania\\
I. U. F., Paris, France\\}

\end{document}